\documentclass[preprint,12pt]{elsarticle}

\journal{arXiv}

\usepackage{amsmath,amsthm,amsfonts,amssymb,bbm,mathtools}
\usepackage[title]{appendix}
\usepackage{dsfont,extarrows,enumerate}
\usepackage{fullpage}
\usepackage[dvipsnames]{xcolor}
\usepackage[compatibility=false]{caption}
\usepackage{subcaption}
\usepackage[linesnumbered,ruled,vlined]{algorithm2e}
\usepackage[pdfborder={0 0 0},pdfpagemode=UseOutlines,bookmarks=false,breaklinks=true,colorlinks=true,linkcolor=blue,citecolor=blue,urlcolor=blue, pdfauthor=author]{hyperref}
\usepackage{graphicx}
\usepackage[normalem]{ulem}

\numberwithin{equation}{section}

\theoremstyle{plain} \newtheorem{theorem}{Theorem}[section]
\theoremstyle{plain} \newtheorem{proposition}[theorem]{Proposition}
\theoremstyle{plain} \newtheorem{lemma}[theorem]{Lemma}
\theoremstyle{plain} 
\theoremstyle{definition} 
\theoremstyle{definition} 
\theoremstyle{remark} 
\theoremstyle{remark} \newtheorem{remarks}[theorem]{Remarks}
\theoremstyle{remark} \newtheorem{algorithm2}[theorem]{Algorithm}
\theoremstyle{remark}

\newcommand{\E}{\mathbb{E}}
\newcommand{\PP}{\mathbb{P}}

\newcommand{\R}{\mathbb{R}}

\RestyleAlgo{ruled}

\SetKwComment{Comment}{$\triangleright$\ }{}

\usepackage{dcounter} 
\countstyle{section}%
\DeclareDynamicCounter{algocf}

\AtBeginDocument{%
  \def\thefnote{\myfnsymbol{fnote}}}
\makeatletter
\def\myfnsymbol#1{\expandafter\@myfnsymbol\csname c@#1\endcsname}
\def\@myfnsymbol#1{\ifcase #1\or ***\else \@ctrerr\fi}
\def\fntext[#1]#2{\g@addto@macro\@fnotes{%
   \refstepcounter{fnote}\elsLabel{#1}%
   \def\thefootnote{\thefnote}
   \global\setcounter{footnote}{\c@fnote}%
   \footnotetext{#2}}}
\makeatother

\begin{document}

\begin{frontmatter}

\title{{\bf \small Analytic and Monte Carlo Approximations to the Distribution of the First Passage Time of the Drifted Diffusion with Stochastic Resetting and Mixed Boundary Conditions}
}

\author[1,2]{\small Juan Magalang \corref{cor1}}
\author[1,3]{\small Riccardo Turin \corref{cor1}}
\author[4]{\small Javier Aguilar \corref{cor2}}
\author[1]{\small Laetitia Colombani}
\author[2]{\small Daniel Sanchez-Taltavull \fnref{foot1}}
\author[1]{\small Riccardo Gatto \fnref{foot1}}

\affiliation[1]{organization={Institute of Mathematical Statistics and Actuarial Science, University of Bern},
            addressline={Alpeneggstrasse 22}, 
            city={Bern},
            postcode={3012},
            country={Switzerland}}

\affiliation[2]{organization={Department of Visceral Surgery and Medicine, Inselspital, Bern University Hospital, University of Bern},
            addressline={Murtenstrasse 35}, 
            city={Bern},
            postcode={3008},
            country={Switzerland}}

\affiliation[3]{organization={Group Risk Management, Swiss Re Management Ltd},
            addressline={Mythenquai 50/60}, 
            city={Zurich},
            postcode={8022},
            country={Switzerland}}

\affiliation[4]
{organization={Laboratory of Interdisciplinary Physics, Department of Physics and Astronomy
“G. Galilei”, University of Padova},
            city={Padova},
            country={Italy}}
            
\fntext[foot1]{Scientific group leaders}
\cortext[cor1]{Equal contribution}
\cortext[cor2]{Corresponding author: {\tt javier.aguilarsanchez@unipd.it}}

\begin{abstract}
{\small
This article introduces two techniques for computing the distribution of the absorption or first passage time 
of the drifted Wiener diffusion subject to Poisson resetting times, to an upper hard wall barrier and to a lower absorbing barrier. The first method,
which we call ``Pad\'e-partial fraction'' approximation, starts with the Pad\'e approximation to the Laplace transform of the 
first passage time, which is then exactly inverted by means of the partial fraction decomposition. The second method, which we call
``multiresolution algorithm'', is a Monte Carlo technique that exploits the properties of the Wiener process in order to generate Brownian bridges at increasing levels of resolution. 
Our numerical study reveals that the multiresolution algorithm has higher efficiency than 
standard Monte Carlo,
whereas the faster Pad\'e-partial fraction method
is accurate in various circumstances and provides
an analytical formula.
Also, a closed-form exact expression for the expected 
first passage time is derived.}
\\

\end{abstract}

\begin{keyword}
absorbing barrier 
\sep Brownian bridge 
\sep Euler-Maruyama scheme
\sep Fokker-Planck equation \sep hard wall reflection \sep Laplace transform  
\sep Pad\'e approximation 
\sep partial fractions decomposition
\sep Poisson resetting times 
\sep Wiener process.
\PACS 02.50.-r \sep 05.10.Gg \sep 05.10.Ln \sep 05.40.-a \sep 02.30.Mv \sep 02.60.-x \sep 02.50.Ng
\MSC 41A58 \sep 60G18 \sep 60G40 \sep 60J65 \sep 65C05 \sep 35G15
\end{keyword}

\end{frontmatter}

\tableofcontents

\newpage

\section{Introduction}

Drift-diffusion processes represent a cornerstone in the mathematical modeling of systems whose evolution includes stochastic components. 
The central one is Brownian motion and it finds frequent applications in most scientific fields, such as physics~\cite{Kampen_2007, bian2016111, libchaber2019biology, Romanczuk2012}, biology~\cite{Romanczuk2009, Woods2014, Marcovitz2013, Oshanin2009},
insurance mathematics~\cite{dufresne1991risk,gatto2016saddlepoint, gatto2012four} and it builds the basis of mathematical finance.
Through this work, we will focus on a drifted and bounded Brownian motion for which the particle returns to its original position at random times following the Poisson process. Such stochastic resetting has gained importance in the past decade \cite{Majumdar2015, evans_diffusion_2011, chechkin2018searchresettingunifiedrenewal,evans_stochastic_2020, pal_first_2019, roldan_path-integral_2017,nagar_diffusion_2016,pal_diffusion_2016,Kusmierz2014,Meylahn2015}. Particularly, in the context of biology, stochastic resetting has been used to model hunting behavior in animals, where animals return to specific sites to look for food \cite{mercado2018lotka}. At the cellular level, this scheme can model cellular focal adhesions \cite{bressloff2020stochastic}. At the biomolecular level, it has been applied to model backtrack recovery in RNA polymerization \cite{roldan2016stochastic}.  In our previous work, we have shown that stochastic resetting can be used to model the role of changes in therapies to palliate drug resistance development \cite{ramoso_stochastic_2020}. As new models and experimental realizations of processes with stochastic resetting continue to emerge~\cite{besga2020optimal,tal2020experimental}, there is a growing need for more refined approximation and simulations techniques to comprehensively characterize real-world resetting protocols~\cite{evans_stochastic_2020}, which we tackle to some extent in this study.

By means of the Laplace transform and of the
underlying stochastic differential equation (SDE) \cite{pavliotis_stochastic_2014, karlin_taylor_1981},
we are able to derive an analytic and a Monte Carlo
technique for computing the distribution of the 
first passage time (FPT) to the null level of the 
process~\cite{redner2001guide,iyer2016first,aguilar2023endemic,zhang2016first,Chicheportiche2014,Masoliver2014,Maggiore2014,Barkai2014}.
With other methods existing in the literature the mean FPT can be obtained analytically. However, this first moment does not provide sufficient information the entire FPT distribution: quantities such as standard deviation, median, and upper quantiles are often relevant in applications to biology and medicine, for example. 

In this context, this article introduces two computational methods for obtaining the probability distribution of the FPT of the drifted Brownian motion subject to Poisson resetting times and to upper hard wall barrier. The first of these two techniques makes use of the Laplace transform of the FPT~\cite{ramoso_stochastic_2020}. It is difficult to find a general and accurate method for inverting Laplace transforms.
Moreover, most methods are purely numerical. Accurate analytical approximation formulae are however useful in various situations, for example for computing sensitivities of the approximated probabilities.
Our method uses the Pad\'e approximation and partial fractions decomposition to approximate the Laplace transform inversion. We call it Pad\'e-partial fraction (PPF) approximation.
Besides high computational speed, it provides a simple closed-form
expression for the distribution of the FPT.
The second method is a Monte Carlo algorithm that
exploits a bridge property of the Wiener process, which
allows to obtain trajectories at increasing level of detail \cite{Davies1979, asmussen2007stochastic}. We call it multiresolution algorithm (MRA), following the terminology of wavelet
analysis. The MRA allows for very high accuracy.
We introduce two versions of the MRA: the standard MRA (SMRA), which directly exploits the bridge property of the Wiener process to our model, and the hybrid MRA (HMRA), which starts with the classical Euler-Maruyama algorithm to generate the initial approximation of the FPT and which improve it further by using the MRA. 

This article has the following structure. 
We first derive the Laplace transform of the FPT (Section \ref{s2}). Then
we introduce the PPF approximation (Section \ref{sec:approxinv}). 
This section presents also a simple
expression of the mean FPT.
Next, we present the SMRA and HMRA
(Section \ref{sec:multires}). These methodological
sections are followed by an intensive numerical
study, aiming to show the accuracy of our techniques
(Section \ref{num}).
A short summary of the methodological and numerical
results followed by a discussion on future research 
concludes the article (Section \ref{sec:discussion}).
\\

\noindent {\bf Notation} \\
\noindent
$\bullet$ $f^{(k)}(x)= (d/dx)^k f(x)$, for $k=0,1,\ldots$, where $f: \R \to \R$. \\
\noindent
$\bullet$ $\partial_t{f} (t,x) = (\partial / \partial t)
f(t,x)$,
$f^\prime (t,x)= (\partial / \partial x)
f(t,x)$ and $f^{\prime\prime} (t,x) = (\partial / \partial x)^2
f(t,x)$,
where 
$f: [0,\infty) \times \R \to \R$.
\\
\noindent $\bullet$
    $\tilde{f}(s)
	=\int_0^\infty e^{-st} f(t)\,dt$
 is the Laplace transform of $f:[0,\infty) \to \R$.
 \\
 \noindent $\bullet$ $\delta(x)$ is the Dirac delta function, which assigns mass 1 at $x=0$ and is null $\forall x \neq 0$. \\
\noindent
$\bullet$ $X \sim Y$ means that the random variables $X$
 and $Y$ have same distribution. \\
 \noindent $\bullet$ $f(x)=o\left(g(x)\right)$, as $x\to a$, means that $\lim_{x\to a} \frac{f(x)}{g(x)}=0$, where $f,g: \R \to \R$.\footnote{It is mplicitely assumed that
 $g$ is nonnull over some neighborhood of $a$, if $|a| < \infty$, or for all extreme large (small) values, if $a=\infty$ ($-\infty$).}

\section{Laplace transform of propagator and mean FPT}                 \label{s2}
The distribution of FPTs is often the essential 
element in the study of absorption phenomena. This distribution crucially depends on the drift-diffusion processes, the number of dimensions of the diffusion space, and the boundary conditions~ \cite{redner2001guide, yustelindenberg, dyesguerra, voiturezbenichou}. Of particular interest are problems with absorbing or reflecting boundaries. Absorbing boundaries are regions in the diffusion space where diffusing particles can enter but cannot leave~\cite{cox65theory}. Reflecting boundaries are regions in which diffusing particles cannot permeate \cite{Skorokhod1961}. In the following, we will focus on problems with mixed boundary conditions, meaning that the process is bounded between one absorbing and one reflecting boundary.

Obtaining the closed-form expression for the distribution of the FPT is in general difficult, because it involves solving a Fokker-Planck equation with absorbing boundaries~\cite{redner2001guide, dyesguerra, voiturezbenichou}. Nonetheless, we are still able to characterize the FPT through the Laplace transform. In this section, we derive the Laplace transform of the FPT distribution and obtain a novel formula for its expectation.
Starting with the process without resetting, we obtain the Laplace transform of the propagator by solving the Laplace transform of the Fokker-Planck equation with boundary conditions \cite{cox65theory} (Section \ref{subsec:2-1}). From the propagator, we obtain the Laplace transform of the FPT distribution \cite{redner2001guide, ramoso_stochastic_2020} that will relate to the process with resetting through the survival function \cite{pal_first_2019, roldan_path-integral_2017} (Section \ref{subsec:2-2}). Finally, by using the aforementioned expressions, we show that the expected FPT with resetting can be expressed in terms of the Laplace transform of the FPT distribution without resetting.

\subsection{Laplace transform of the propagator}\label{subsec:2-1}

The underlying stochastic process is the basic
Brownian motion or Wiener process with drift, here denoted $Y=\{Y_t\}_{t \ge 0}$, which solves the SDE 
\begin{equation}
d Y_t=v\,dt+\sqrt{2D}\,d W_t,\; \forall t\geq 0\,,
 \label{eq:sde}
\end{equation}
with fixed initial condition $Y_0=x_0\in(0,1]$, where the drift 
$v$ and the volatility $D$ are respectively real
and nonnegative real constants.
We then impose that the paths of this process are bounded between $0$ and $1$ and use the same name $Y$ for the bounded process. The level 1 is a \emph{reflecting boundary of type hard wall}~ \cite{redner2001guide}. On the one hand, the definition of this reflecting boundary at the level of the Fokker-Planck equation requires the concept of probability current and will be given in Eq.~\eqref{eq:bc}. On the other, the action of the reflecting boundary at the level of sample path will be provided along with its discretization scheme in Section~\ref{subsubsec:4-1-1}. We refer to~\cite{pilipenko2014introduction} for a detailed discussion of SDEs with reflecting boundaries.

We note that this process is neither the formal reflection of a trajectory, as would be the absolute value of a diffusion, nor the regulated Brownian motion. These two cases are described in the introduction of \cite{harrisonBrownianMotionStochastic1985}. 
Regarding the null level, it is an \emph{absorbing state} and our
goal is precisely to evaluate the probability of reaching this state.

Our prior work \cite{ramoso_stochastic_2020} showcases the practical 
relevance of the model in describing biological phenomena. Specifically, we utilized a drift-diffusion process with mixed boundary conditions in order to model drug resistance development resulting from mutation -- a stochastic process biased towards the survival of the infecting pathogen~\cite{MANRUBIA2006, Pillay1998}. Our approach employed a bounded drift-diffusion process to quantify changes in therapy efficacy against the mutating infecting pathogen. The reflecting boundary represented perfect therapy, while the absorbing boundary denoted drug failure, since infecting pathogens develop complete resistance to therapies due to mutation~\cite{Clutter2016}.

Denote by $p(\cdot,t)$
the probability density of $Y_t$, $\forall t >0$,
and $p(\cdot,0)= \delta(\cdot-x_0)$, called \emph{propagator}.
Let $x \in \R$ and $t \ge0$.
The {\it forward Kolmogorov} or {\it Fokker-Planck equation} determines
the probability distribution of the process $Y$ and it is given by the PDE
\begin{equation}
\label{eq:fke}
	\partial_t{p}(x,t)+vp'(x,t)-Dp''(x,t)=0.
\end{equation}
We refer e.g. to \cite{cox65theory}
for the construction of Eq. (\refeq{eq:fke}).

The propagator is indeed a defective probability density, viz.
with total mass below 1, because it does not account for the probability mass at the absorption state.
This reflects the physical interpretation of $Y$ as the location of a particle moving between
two boundaries that as soon as touches the absorbing boundary gets immediately removed from
the system.

Let us further introduce the \emph{probability current}
\begin{equation} \label{eq:probcurrent}
	J(x,t)=vp(x,t)- Dp'(x,t)\,.
\end{equation}
The two boundary conditions that are  
absorption at 0
hard wall at 1 can be weakly defined in terms of the propagator and the probability current in the following way,
\begin{equation}
\label{eq:bc}
	p(0,t) = 0 \quad \text{and} \quad J(1,t)=0,
\end{equation}
respectively;
see e.g. Chapter 4 of \cite{pavliotis_stochastic_2014} or \cite{redner2001guide}.

Thus we study the solution of the following system of equations, 
\begin{align} \label{eq:system_of_p}
\left\{
				\begin{array}{ll}
					\partial_t{p}(x,t)+vp'(x,t)-Dp''(x,t)=0 , \\
    p(0,t) = 0, \\  J(1,t)=0, \\
    p(x,0) = \delta(x-x_0),
				\end{array}
				\right.
\end{align}
in which the last equation provides the initial condition.
No closed-form solution to Eqs. \eqref{eq:system_of_p} is available. Lemma \ref{lemma:propagator} 
provides a closed-form expression for
the Laplace transform of the propagator, viz. for
$\tilde{p}(x,s)
	=\int_0^\infty e^{-st}p(x,t) \,dt$, $\forall s\geq 0$.
\begin{lemma}[Laplace transform of propagator]
\label{lemma:propagator}
Let $x_0 \in (0,1]$, $v \in \mathbb R$, $D>0$ and 
consider $p$ solution to Eqs. \eqref{eq:system_of_p}.
Let $s> - v^2/(4D)$ and define 
\begin{align}                           \label{eq:constants}
\rho & = \dfrac{v}{2D},
\quad \omega(s)=\sqrt{v^2+4 D s}, \quad
\theta(s)=\dfrac{\omega(s)}{2D}=\dfrac{\sqrt{v^2+4 D s}}{2D}
\end{align}
and
\begin{equation} \label{eq:roots}
	\alpha_\pm(s)
	=\frac{v\pm\sqrt{v^2+4Ds}}{2D}
	=\frac{v\pm\omega(s)}{2D}
    =\rho\pm\theta(s).
\end{equation}
Then the Laplace transform of the propagator 
is given by
\begin{align}\label{eq:p_below}
	\tilde{p}(x,s) & = \tilde{p}(x_0,s) \times
     \begin{cases}
	\frac{e^{\alpha_+(s) x}-e^{\alpha_-(s) x}}
	{e^{\alpha_+(s) x_0}-e^{\alpha_-(s) x_0}}, &
	\forall x\in[0,x_0),   \\
	\frac{\alpha_+(s) \, e^{\alpha_+(s) (x-1)} - \alpha_-(s) \, e^{\alpha_-(s) (x-1)}}
	{\alpha_+(s) \, e^{\alpha_+(s) (x_0-1)} - \alpha_-(s) \, e^{\alpha_-(s) (x_0-1)}}, & \forall x\in(x_0,1],  
\end{cases}
\end{align}
with $\tilde{p}(x_0, s)$ given by
\begin{align} \label{eq:px_0}
	\tilde{p}(x_0, s)
	&=\frac{2 \, \mathrm{sinh}(\theta x_0)\left\{\omega\,\mathrm{cosh}[\theta(x_0-1)]+v\,\mathrm{sinh}[\theta(x_0-1)]\right\}}{\omega^2 \,\mathrm{cosh}(\theta)-v\,\omega\,\mathrm{sinh}(\theta)}\, . 
\end{align}
\end{lemma}
The ratio $\rho$ is currently called {\it P\'eclet number}.

\subsection{Stochastic resetting and stopping condition} \label{subsec:2-2}

We now modify the dynamics of the stochastic process $Y$ by putting
stochastic resetting.
More explicitly, we assume that at random times 
$T_0 = 0 < T_1 < T_2 < \ldots$, a.s., the value of the process $Y$
is reset to its initial value $x_0$.
Following \cite{evans_stochastic_2020},
this is expressed through the addition of a new term to the SDE,
giving 
\begin{equation} \label{eq:sde_with_resets}
	dX_t=(1-\chi_t) \cdot (v\,dt+\sqrt{2D}\,dW_t) +\chi_t \cdot (x_0-X_t),\quad \forall t\geq 0,
\end{equation}
where
\begin{equation*}
	\chi_t =\sum_{n=1}^\infty \mathbf{1} \{T_n = t\}, \quad \forall t \ge 0,
\end{equation*}
where 
$\mathbf{1}$ denotes the indicator.
We assume that the stochastic process $\chi=\{\chi_t\}_{t\ge 0}$ is an independent homogeneous Poisson process with rate or intensity $r>0$. Thus
$T_n$ the sum of $n$ independent exponential random variables with expectation $r^{-1}$, for
$n=1,2\ldots$.

We can now define the 
FPT (also called absorption time) by 
\begin{equation}\label{eq:fpt_rv}
	\tau_r=\inf\{t\geq 0\, | \, X_t\leq0\}.
\end{equation}
It admits a proper probability density function denoted $f_r$,
where the subscript $r \ge 0$ highlights the 
dependence on the Poisson rate $r$.
When $r=0$, we retrieve the dynamic without resetting.
 
\section{Laplace transform of FPT distribution and PPF} \label{sec:approxinv}
Although the mean FPT $\E[\tau_r]$ properly characterizes typical absorption times, this quantity alone does not quantify the uncertainty inherent to the first passage phenomenon. Often, high-order quantiles are more relevant.
Thus, 
besides the expectation
we want to obtain the entire probability 
distribution of $\tau_r$.
Its Laplace transform
is available, but there is no obvious and 
general way of inverting it.
The well-known fast Fourier transform (FFT) is a purely numerical
method and it does not necessarily provide accurate results, in particular for approximating
upper tail quantiles, that are useful in many applications.
We refer for example to 
\cite{gatto2012four} for
a numerical comparison of methods 
for computing a FPT probability 
for the compound Poisson process perturbed by diffusion.

In this section, we propose and implement a particular method for our FPT problem. We obtain the Laplace transform of the FPT 
(Section \ref{subsec:3-1}), PPF approximation whose inversion will be approximated numerically. The approximated inversion begins with the Pad\'e approximation of the Laplace transform, followed by a partial fraction decomposition of the Pad\'e approximation and then by the simple inversion of the sum of partial fractions (Section \ref{subsec:3-2}). 
This is the PPF
approximation. One can find
 some references, in particular in the engineering literature, on the problem of obtaining approximate Laplace inversions by rational approximations: some early references are \cite{Davies1979}, \cite{luke1961approximate}, \cite{ longman1973generation}, \cite{akin1969application} and \cite{luke196977ie}.

\subsection{Laplace transform of FPT distribution with resetting} \label{subsec:3-1}

In this section we first provide the Laplace transform of the FPT distribution, 
$\tau_r$, and then we give a simplified formula for its expectation. 

\begin{proposition}[Laplace transform of FPT distribution with resetting] \label{prop:laplaceres}
    Assume that $X$ solves the SDE with resetting of Eq. \eqref{eq:sde_with_resets}, with $X_0 = x_0 \in (0,1]$, $v \in \mathbb R$, $D>0$ and $r>0$. 
    Denote by $f_r$ the probability density of the absorption time $\tau_r$, defined in Eq. \eqref{eq:fpt_rv}. Then its Laplace transform is given by, $\forall s >-r$,
    \begin{align} \label{eq:tildef}
\tilde{f}_r(s) &
    = \frac{(s+r)\tilde{f}_0(s+r)}{s+r\tilde{f}_0(s+r)} \nonumber \\
    & 
    {\textstyle
    =\frac{(s+r) \left\{ \omega(s+r) \, \mathrm{cosh}[\theta(s+r)(x_0-1)]+v \, \mathrm{sinh}[\theta(s+r)(x_0-1)]\right\}}
    {s e^{\rho x_0}\left\{ \omega(s+r) \, \mathrm{cosh}[\theta(s+r)]-v \, \mathrm{sinh}[\theta(s+r)] \right\}
    + r \left\{ \omega(s+r) \, \mathrm{cosh}[\theta(s+r)(x_0-1)]+v \, \mathrm{sinh}[\theta(s+r)(x_0-1)] \right\} }},
\end{align}
where $\tilde{f}_0$ is given by, $\forall s> -v^2/(4D)$,
\begin{equation} \label{eq:tildef0}
	\tilde{f}_0(s)=
	e^{-\rho x_0}\frac{\omega(s) \, \mathrm{cosh}\{\theta(s)(x_0-1)\}+v \, \mathrm{sinh}\{\theta(s)(x_0-1)\}}
	{\omega(s) \, \mathrm{cosh}\{\theta(s)\}-v \, \mathrm{sinh}\{\theta(s)\}}\,.
\end{equation}
\end{proposition}
\begin{proof}
    We obtain from Eq.~\eqref{eq:kr} the Laplace transform of $f_r$ as
$\tilde{f}_r(s)=1-s\tilde{S}_r(s)$, for all $s \in \mathbb{R}$.
We combine it with Eq.~\eqref{eq:tildeSr} and it yields the first 
expression,
\begin{align*}
\forall s > -r, 	\; \tilde{f}_r(s)=\frac{(s+r)\tilde{f}_0(s+r)}{s+r\tilde{f}_0(s+r)}
 \; .
 \end{align*}
 This last expression holds $\forall s < -r$ under the assumption $ \tilde{f}_0(s+r)\neq -s/r$.  

We deduce the closed-form expression 
 in Eq.~\eqref{eq:tildef}
 by combining the first expression with $\tilde f_0$ in Eq.~\eqref{eq:tildef0}. Note that, for the function $\omega$ to be defined, we need $s+r \in \left(-v^2/(4D),\infty \right)$, viz. $s \in \left(- r -v^2/(4D), -r \right)$. 
\end{proof}

We can make two short remarks on Proposition
\ref{prop:laplaceres}. First, because the Laplace
transform exists over a neighborhood of the 
origin, all moments of $\tau_r$ exist and
determine its distribution unambiguously.
In contrast with this, the FPT of the driving drifted Brownian motion without hard wall reflecting boundary and resetting
has infinite moments:
it is known that although the FPT of the Brownian
motion is finite with probability one, its 
expectation is infinite~\cite{redner2001guide}.

We also note the distribution of the model without resetting
is obtained continuously, because
$\lim_{r \to 0} \tilde{f}_r(s) = \tilde{f}_0(s)$,
$\forall s>0$.

We end this section with the following novel closed-form expression for the
\emph{mean absorption time} with resetting.
\begin{proposition}[Mean FPT]       \label{p22}
Assume that the process $X$ solves the SDE with resetting of Eq.~\eqref{eq:sde_with_resets}, with $X_0 = x_0 \in (0, 1]$, 
$v \in \mathbb R$, $D>0$ and $r>0$. The expectation of $\tau_r$ in Eq.~\eqref{eq:fpt_rv} is given by
\begin{equation}
	\mathbb{E} [\tau_r]
	=\frac{1}{r}\left(\frac{e^{\rho x_0} \, (\omega(r) \, \mathrm{cosh}\{\theta(r)\}-v \, \mathrm{sinh}\{\theta(r)\})}{\omega(r) \, \mathrm{cosh}\{\theta(r)(x_0-1)\}+v \, \mathrm{sinh}\{\theta(r)(x_0-1)\}}-1\right),
\end{equation} 
where $\rho$, $\omega$ and $\theta$ are given in Eq.~\eqref{eq:constants}. 
\label{prop:mfpt}
\end{proposition}
The proof of Proposition ~\ref{p22} is in Appendix \ref{ap:proof}.

\subsection{PPF approximation to FPT distribution with resetting} \label{subsec:3-2}
Because we cannot invert the Laplace transform 
$\tilde{f}_r$ in
Eq.~\eqref{eq:tildef} analytically, we now
propose to approximate $\tilde{f}_r$ by 
a specific rational function, which is
the ratio of two polynomials with degree in denominator
higher than in numerator. The type of rational approximation
considered here is the one suggested by 
Henri Pad\'e thesis in 1892 and called
Pad\'e approximation. It 
has shown practical relevance in many problems of
theoretical physics where power series expansions occur,
as already stressed by \cite{baker1961pade}. 
Introductions can be found in various books, such as in 
Section 4.6 of \cite{mathews2004numerical}, in which it is mentioned
that for given computational time, the Pad\'e approximation
is typically more accurate than the Taylor approximation.
It has a local error of order smaller than the sum of the two
degrees of the two polynomials of the ratio. We note that available tables do not provide the analytical form of the Laplace inverse of a Taylor approximation and, under some assumptions, they provide the Laplace inverse of many rational functions. 

    Let $m,n$ be nonnegative integers, let $I$ be an interval of $\mathbb{R}$ with $0 \in I$ and let $g : I \rightarrow \mathbb{R}$ be a $n+m$ times differentiable function.
 Then the Pad\'e approximation of orders 
 $m$ and $n$ to $g$ is the rational function 
    \begin{equation}
    \label{eq:pade}
	g_{m,n}(s)=\frac{p_m(s)}{q_n(s)},
 \end{equation}
 where
 $$p_m(s) = \sum_{j=0}^m a_j s^j \text{ ,and } 
 q_n(s) = \sum_{j=0}^n b_j s^j, \quad \forall s \in I,$$
 and 
 which satisfies $g^{(k)}(0) = g_{m,n}^{(k)}(0)$,
    for $k = 0, \ldots, n+m$.\footnote{Note that these same
    equations hold for the Taylor approximation
of order $m+n$.}
It can be shown that
these conditions imply that 
the coefficients $a_0, \ldots, a_m$ and $b_0,\ldots,b_{n-1}$ are 
obtained by solving the $m+n+1$ linear equations
\begin{align}                       \label{coeff}
    \left\{
    \begin{aligned}
        &\sum_{j=0\vee (i-n)}^i \frac{g^{(j)}(0)}{j!} b_{i-j} = a_i, \text{ for } i=0,\dots,m,
        \\
        &\sum_{j=0\vee (i-n)}^i \frac{g^{(j)}(0)}{j!} b_{i-j} = 0, \text{ for } i=m+1,\dots,m+n.
    \end{aligned}
    \right.
\end{align}
It can also be shown that
\[
    g(s)-g_{m,n}(s)=o(s^{m+n}),\quad\text{as}\quad s\to 0\,.
\]

The PPF considers the restriction $m<n$. It
allows for exact inversion of the Laplace transform in Eq. 
(\ref{eq:tildef})
approximated by Pad\'e and re-expressed in terms of partial 
fractions.
This leads to a simple and practical approximate 
closed-form expression for the 
density of the FPT $\tau_r$. Precisely,
let $g_{m,n}(s) =p_m(s)/q_n(s) $ be the 
Pad\'e approximation of orders $m<n$ to the Laplace transform $\tilde{f}_r$. If
{\it all roots of $q_n$ are real and negative}, then there exists a Laplace inversion of $g_{m,n}$, which we denote $f_{m,n}$. 
Thus, $f_{m,n}$ provides an approximation to the density of $\tau_r$, namely to $f_r$ on $[0, \infty)$.
Thus the density $f_{m,n}$ can be obtained by the following steps, under the
above condition on roots of $q_n$.
\begin{algorithm2}[PPF approximation to FPT density]
    \label{prop:Pade_approximant}
    $\;$ \\
\noindent \textbf{Step 1.} Pad\'e approximation
\begin{itemize}
\item
Select the orders $m<n$ of the Pad\'e approximation 
to $\tilde{f}_r$, which we denote $g_{m,n}$.
\item Compute $\tilde{f}_r^{(k)}(0)$, for $k = 0, \ldots , n+m$.
\item Find 
the coefficients $a_0, \ldots, a_m$ and $b_0,\ldots,b_{n-1}$
by solving the $m+n+1$ linear equations in Eqs.
(\ref{coeff}).
\end{itemize}
\noindent \textbf{Step 2.} Partial fraction decomposition
\begin{itemize}
\item
Decompose the Pad\'e approximation 
in the partial fractions as follows,
\begin{equation}
\label{eq:pfd}
g_{m,n}(u) = 
	 \frac{p_m(u)}{q_n(u)}=\sum_{j=1}^k \sum_{i=1}^{l_j} \frac{\gamma_{j,i}}{(u-\alpha_j)^i},
\end{equation}
where $\alpha_1, \ldots, \alpha_k$ are the distinct roots of the denominator $q_n$, whose multiplicities respectively are 
$l_1, \ldots, l_k$ (with $l_1 + \ldots + l_k = n$).
\item Compute the real coefficients $\gamma_{j,i}$, 
for $j = 1, \ldots, k$ 
and $i = 1, \ldots, l_j$.
Typically, $q_n$ has $n$ distinct roots $\alpha_1,\dots,\alpha_n$.
In this case,
\begin{equation*}
	 \frac{p_m(u)}{q_n(u)}=\sum_{j=1}^n \frac{\gamma_{j}}{u-\alpha_j},
\end{equation*}
gives us
\begin{equation*}
    \gamma_j=\frac{p_m(\alpha_j)}{\prod_{k=1, k\neq j}^n (\alpha_j-\alpha_k)}\,,\quad \text{ for } j=1,\dots, n\,. 
\end{equation*}

Should not all roots of $q_n$ be distinct, the coefficients $\gamma_{j,i}$, for $j = 1, \dots, k$ 
and $i = 1, \dots, l_j$, could be obtained by 
the residue method.
\end{itemize}
\noindent \textbf{Step 3.} Inversion \\
This step is possible only under the restriction $\alpha_1,\ldots,\alpha_k<0$. 

\begin{itemize}
\item Invert the Pad\'e approximation 
$g_{m,n}$ to $\tilde{f}_r$ by exploiting its re-expression in Eq. (\ref{eq:pfd}) so to obtain the linear combination of gamma or
exponential densities
\begin{equation} \label{eq:approxform}
    f_{m,n}(t)=\sum_{j=1}^k \sum_{i=1}^{l_j} \frac{\gamma_{j,i}}{i!} t^{i-1} e^{\alpha_j t}\,,\quad \forall t> 0.
\end{equation}
\end{itemize}
\noindent \textbf{Step 4.} Corrections \\
The following corrections to Eq.
\eqref{eq:approxform} generally improve its accuracy. 
In order to keep notation simple, the same name $f_{m,n}$ is
used for the original approximation of Eq.
\eqref{eq:approxform}
and for the corrected version.
\begin{itemize}
\item {\it Truncation of negative parts}\\
Negative values are equated to null, viz.
$$f_{m,n}(t) = \max \{f_{m,n}(t), 0\}, \; \forall t >0.$$
\item {\it Smoothing near to origin} \\
Oscillations near to the origin are removed.
Find $s>0$ the abscissa point of the last local minimum of $f_{m,n}$.
If it does not exist, then no correction is required.
Consider 
\begin{equation}                                \label{cat}
f_{m,n}(t) = 
\begin{cases}
    0,          & \text{if } t \le s, \\
    f_{m,n}(t), & \text{if } t>s.
\end{cases}
\end{equation} 
\item {\it Normalization} \\
Give integral value one, viz.
consider $$
f_{m,n}(t) = 
\left( \int_0^\infty f_{m,n}(s) ds \right)^{-1} f_{m,n}(t), \; \forall t >0.$$
\item {\it Recentering to expectation} \\
Give expected value $\mu = \mathbb{E} [\tau_r]$.  Compute $\mu$,
through its closed-form expression
in Eq.~\eqref{prop:mfpt}, and compute $\mu_{m,n}$, 
the mean of the approximate density $f_{m,n}$. Obtain
    the recentered approximated FPT density by 
    $$f_{m,n}(t) = f_{m,n}(t + \mu_{m,n} - \mu), \; \forall t > 
     \max\{0,\mu - \mu_{m,n}\}.
    $$
\end{itemize}
\end{algorithm2}

Let us now provide some further and more precise justifications to the
PPF Algorithm \ref{prop:Pade_approximant}.
The PPF approximation $f_{m,n}$ to the FPT density $f_r$ is
implicitly defined through its Laplace inversion formula
$g_{m,n}(u)=\int_0^\infty e^{-ut} f_{m,n}(t)\ dt$,
    $\forall u \ge 0$.
If $\alpha_1,\ldots,\alpha_k<0$, we have
\begin{equation*}
    \int_0^\infty \frac{t^{i-1}}{i!}  e^{(\alpha_j-u)t}\ dt=\frac{1}{(u-\alpha_j)^i},\;\forall  u > \max\{\alpha_1,\ldots,\alpha_k\},
\end{equation*}
for $i= 1, \dots, l_j$ and 
    $j=1,\dots, k$, 
    This and the partial fractions decomposition 
    of $g_{m,n}$
    in Eq.~\eqref{eq:pfd}
give Eq. (\refeq{eq:approxform}), in Step 3.

Next, the two corrections regarding
truncation to nonnegative values and
smoothing near the origin
concern only small neighborhoods of the origin.
The corrections are due to the fact that the PPF approximation can display
undesirable oscillations 
over these neighborhoods. However, we know that the true FPT density must vanish at the origin. 
Indeed, the process starts above zero. In addition to this,
repeated Monte Carlo experiments have confirmed that the FPT 
density is smooth close to the origin and overall unimodal.
Based on this theoretical or empirical evidence, 
the above smoothing correction assumes that the domain of the 
density is cut in two parts: the region where it increases from null
to the maximum of the density,
followed by the region where it decreases to null.
Accordingly, we correct near to the origin 
by finding the set of points $t>0$ such that $f_{m,n}'(t) = 0$.
We assume that the largest value in this set is mode of the 
approximate density and so
all other values of the set must be local extrema. If this set has one element only, then no correction is needed.
Otherwise, we denote $t^\ddag$ to be 
the second largest value in this set (which is either the last local minimum before the mode of or the last point of inflection): we remove
the non-monotonic part around the origin through Eq. (\ref{cat}). 

We conclude with the next remarks.
\begin{remarks}
\begin{itemize}
\item[1.]
As mentioned, $f_{m,n}$ is not necessarily a probability density: as $\gamma_1,\ldots,\gamma_n$ can be negative, $f_{m,n}$ can be negative over some regions.
However,
\begin{align*}
    \int_0^{\infty} f_{m,n}(t) dt &= \sum_{j=1}^k \sum_{i=1}^{l_j} \frac{\gamma_{j,i}}{(-\alpha_j)^i} 
    = g_{m,n} (0) \longrightarrow \tilde{f}_r(0) = 1 , 
    \text{ as } m,n \rightarrow \infty,
\end{align*}
so the sequence of PPF approximations has the
correct normalization in the limit of large $m,\,n$. 
\item[2.]
If the PPF approximation $f_{m,n}$ is a proper probability density, then its expectation is given by
\begin{align*}
\mu_{r,m,n}=
    \sum_{j=1}^k \sum_{i=1}^{l_j} \frac{(i+1) \gamma_{j,i}}{(-\alpha_j)^{i+1}}.
\end{align*}
\item[3.] If the function $f_{m,n}$
has a Laplace transform of the form of the Pad\'e
approximation Eq.~\eqref{eq:pade},
i.e. if $f_{m,n}=p_m/q_n$ for some polynomials
$p_m$ and $q_n$, then $f_{m,n}$ 
is characterized as the solution of the homogeneous linear
ordinary differential equation (ODE)
$$
f_{m,n}^{(k)}(t) + c_{k-1} f_{m,n}^{(k-1)}(t) + \ldots +
c_{1} f_{m,n}'(t)
+ c_0 = 0,$$
for some coefficients 
$c_0,\ldots,c_{k-1} \in \mathbb{R}$ 
with $c_0 \neq 0$ and
for some positive integer $k$.
Thus, the closeness of $g_{m,n}$,
obtained in Step 3,
to $f_r$, the true density, can
be re-expressed in terms of closeness of the 
solution of the above 
ODE to the true density.
This may give another way for analyzing the 
error of the PPF approximation.
\item 
If we consider the process without hard wall and without resetting ($Y$), then the FPT follows a simple distribution, which is the inverse Gaussian. In this case, the PPF becomes meaningless.
If we consider the process without 
either hard wall or resetting ($Y$), but not without the two together, then it remains interesting to compute the FPT distribution
and the PPF can be adapted accordingly.
This remarks extends to the Monte Carlo MRA of
Section \ref{sec:multires}.
\end{itemize}
\end{remarks}

\section{MRAs}
\label{sec:multires}

We will see that, although the PPF approximation is efficient and provides analytical formulas, it is only valid for specific combinations of the models parameters. 
Monte Carlo methods are in general a good alternative to compute the distribution of any FPT, however, they are computationally intensive.
In this section we propose a Monte Carlo method that we call the multiresolution algorithm (MRA), which leads to arbitrary accuracy and a combination of the MRA and the Euler-Maruyama algorithm, which reduces its computational time while keeping the accuracy.  

\subsection{MRA for the Wiener process}
\label{subsec:pure}

We describe a particular
strategy for generating trajectories of
Wiener processes, yielding the SMRA.
We provide an algorithm that allows the 
generation of a single sample path
of the Wiener process with arbitrary
resolution. The algorithm is
directly obtained from the following well-known
property of the Wiener process,
\begin{equation}\label{eq:conditioned_wiener}
    \mathrm{P} \Bigl[W_{t+h}\in (x,x+dx) \Big\vert W_t = a, W_{t+2h}=b \Bigr] = \frac{1}{\sqrt{\pi h}}\exp\left\{-\frac{\left(x-{\frac{a+b}{2}}\right)^2}{h}\right\} dx,
\; \forall t,h > 0, a,b \in \mathbb{R}.
\end{equation}

Therefore, given the knowledge of the state of the process at two times ($W_t=a$ and $W_{t+2h}=b$), Eq. \eqref{eq:conditioned_wiener} allows to sample the process at intermediate time ($ W_{t+h}$).  This property can be iterated to 
access arbitrary small temporal scales, as sketched in Figure  \ref{fig:multraj}, 
and we will use it to obtain estimations of FPTs at arbitrary accuracy.
Indeed, consider the stopping time equivalent to Eq. (\ref{eq:fpt_rv}) but for this simpler diffusion, $$T=\inf\{ t \le 0 \mid W_t \le 0\},$$
where we slightly alter the standard definition of the Wiener process to have some initial condition $X_0>0$ that makes $T$ non-trivial. Let us assume that
a discretization 
$\{W_{nh}\}_{n=1,2,\ldots}$ 
of $\{ W_t \}_{t \ge 0}$
is available, for some $h>0$ small. Then,
$$T_{h}=\inf\{ t \ge 0 \mid W_{nh} \le 0\}$$ is
an overestimation of $T$ (i.e. $T_h\ge T$) that converges to
$T$, as $h \to 0$. This tells us that
applying the MRA at increasing resolutions allows us to approximate the stopping time
at any desired precision. 

Given a certain time-interval, for example $t\in[0,1]$, we will note the $k$-th resolution level of the sample path in that interval 
as
$\mathbf{W}_k =  \{ W_{k,j} \}$, for $k=0,1,\ldots$, and  $j=0,1,\dots,2^{k}$. Both indices allows to evaluate the process at a particular time in $[0,1]$, $W_{k,j}=W_{t=2^{-k}j}$. While the first index informs about the level of refinement of the path, the secondary index, $j$, refers to the ordinal of each element given a certain $k$. Therefore, at initial level $k=0$, we have
\begin{equation}
\mathbf{W}_0 = \{W_{0,0} , W_{0,1}\}, \quad \text{where }  W_{0,0} = 0 \text{ and } W_{0,1} \sim \mathcal{N}(0,1).
\label{eq:w00}
\end{equation}
Here $\mathcal{N}(\mu,\sigma^2)$ stands for a Gaussian random variable with mean $\mu$ and variance $\sigma^2$.
Then, consecutive levels of refinement, $k\ge 1$, are obtained by
\begin{subequations}
\label{eq:wiener}
\begin{align}
        W_{k,2j} &= W_{k-1,j}, \quad \text{for } j = 0,1,\ldots, 2^{k-1} \label{subeq:evencase} \\
        W_{k,2j+1} &\sim \mathcal{N}\left(\dfrac{W_{k-1,j}+W_{k-1,j+1}}{2}, \dfrac{1}{2^{k}} \right),  \quad \text{for } j = 0,1, \ldots, 2^{k-1}-1 \label{subeq:oddcase}.
\end{align}
\end{subequations}
 For $k>0$, the even indices of $j$ are copied over as the resolution increases as in Eq.~\eqref{subeq:evencase}. The odd indices of $j$ are generated randomly using two consecutive variables from the previous resolution, as in Eq.~\eqref{subeq:oddcase}. Proceeding in this way, the overestimation of the FPT of $\mathbf{W}_k$ decreases as $k$ increases.

Further details about the MRA can be found at 
p. 277-279 of \cite{asmussen2007stochastic}.  In the Appendix~\ref{AP:multiresolution_algorithm} we show a pseudocode with the implementation of the MRA for the process $\{W_t\}_{t \ge 0}$.

\begin{figure}
    \centering
    \includegraphics[width=0.85\linewidth]{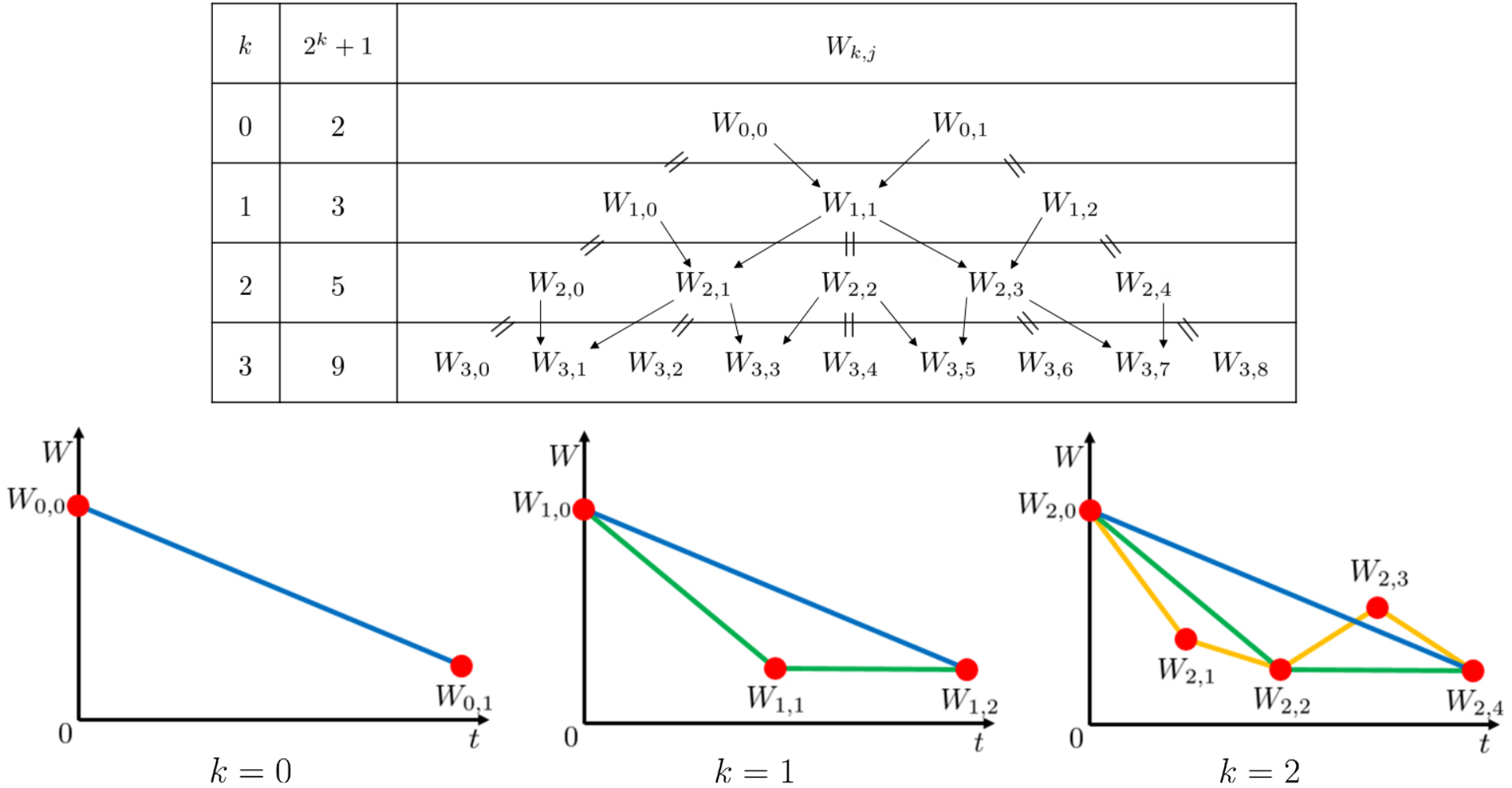}
    \caption{\it Upper table: Resolution level $k$ and corresponding number of generated variables $2^k+1$ (first two columns); generated variables $W_{k,j}$ (last column). Equal signs indicate copying a variable from a previous resolution; see Eq. \eqref{subeq:evencase}. Arrows indicate generating of a new variable from a previous resolution; see Eq. \eqref{subeq:oddcase}. Lower graphs: Illustrations of the trajectory generated subsequently by the MRA.}
    \label{fig:multraj}
\end{figure}

\subsection{MRA with resetting} \label{subsubsec:4-1-2}

This section describes the application of the MRA to a process with resetting, similar to the one of Section~\ref{subsec:2-1}, but without considering the reflecting boundary for now. Let us call $\{B_t\}_{t \ge 0}$ 
the drifted Brownian motion without hard wall and without resetting,
\begin{equation}\label{eq:unbounded_Br}
	dB_t= v\,dt+\sqrt{2D}\,dW_t, \; \forall t \ge 0,
\end{equation}
where $B_0 = x_0>0$. Let us call $\mathbf{t^\dag}$ the sequence of times at which the process experiences a reset

\begin{equation*}
    \mathbf{t^\dag} = \{t_0^\dag, t_1^\dag, t_2^\dag, \ldots\}, \quad \text{where } t_0^\dag=0, \, (t^\dag_{i+1}-t^\dag_i) \sim \text{Exponential}(1/r).
\end{equation*}
Let us define reset intervals from pairs of consecutive reset times: $\{(0, t^\dag_1), (t^\dag_1, t^\dag_2),$ 
$ \ldots \}$. The main idea to exploit is that  $B_{t}$ is a normally distributed random variable with mean $\mu = x_0+vt$ and variance $\sigma^2 = 2Dt$ within any reset interval. Therefore, we can build Brownian bridges connecting the extremes of reset intervals with arbitrary precision. We first consider the process in the first reset interval, $t\in[0,t^\dag_1]$, and, similar to what we did with the Wiener process in Section~\ref{subsec:pure}, we define the $k$-th resolution level of the sample path in that interval as $\mathbf{B}^{(1)}_k$. In the first resolution level, the process starts at the initial position, $B_{0,0} = x_0$, with initial time $0$, and ends at the final position, $B_{0,1}$, with final time $t^\dag_1$. Therefore, 
\begin{equation}
\mathbf{B}_0^{(1)} = \{B^{(1)}_{0,0} , B^{(1)}_{0,1}\}, \quad \text{where } B^{(1)}_{0,1} \sim \mathcal{N}(x_0+vt^\dag_1,
2Dt_1^\dag).
\label{eq:x00}
\end{equation}

Then, similar to Eq.~\eqref{eq:wiener}, the generation of the process at the next resolution $k$ is obtained by
\begin{subequations}
\label{eq:brown}
\begin{align}
        B^{(1)}_{k,2j} &= B^{(1)}_{k-1,j}, \quad \text{for } j = 0,1,\ldots, 2^{k-1},
        \quad \text{and}
        \label{subeq:bevencase} \\
        B^{(1)}_{k,2j+1} &\sim \mathcal{N}\left(\dfrac{B^{(1)}_{k-1,j}+B^{(1)}_{k-1,j+1}}{2}, \dfrac{2Dt^\dag_1}{2^{k}} \right),  \quad \text{for } j = 0,1, \ldots, 2^{k-1}-1. \label{subeq:boddcase}
\end{align}
\end{subequations}
The precise meaning of Eq. (\ref{subeq:boddcase})
is that the {\it conditional} distribution of $B^{(1)}_{k,2j+1}$
{\it given} $B^{(1)}_{k-1,j}$ and $B^{(1)}_{k-1,j+1}$
is Gaussian.
The same steps can be used to generate the Brownian trajectory in the second reset interval, $\mathbf{B}_k^{(2)}$, simply replacing $t_1^\dag$ by $t_2^\dag-t_1^\dag$ in Eqs.~\eqref{eq:x00} and~\eqref{subeq:boddcase}. In general, the Brownian trajectory in the $i^\text{th}$  reset interval, i.e. $(t^\dag_i,t^\dag_{i-1})$, and at refinement level $k$ is obtained through
\begin{equation}
\mathbf{B}_0^{(i)} = \{B^{(i)}_{0,0} , B^{(i)}_{0,1}\}, \quad \text{where } B^{(i)}_{0,1} \sim \mathcal{N}(x_0+v (t_i^\dag-t_{i-1}^\dag),2D(t_i^\dag-t_{i-1}^\dag)),
\label{eq:x00_V2}
\end{equation}
and
\begin{subequations}
\label{eq:brown_V2}
\begin{align}
        B^{(i)}_{k,2j} &= B^{(i)}_{k-1,j}, \quad \text{for } j = 0,1,\ldots, 2^{k-1},
        \quad \text{and}
        \\
        B^{(i)}_{k,2j+1} &\sim \mathcal{N}\left(\dfrac{B^{(i)}_{k-1,j}+B^{(i)}_{k-1,j+1}}{2}, \dfrac{2D(t_i^\dag-t_{i-1}^\dag)}{2^{k}} \right),  \quad \text{for } j = 0,1, \ldots, 2^{k-1}-1. \label{subeq:boddcase_V2}
\end{align}
\end{subequations}

Since $B^{(i-1)}_{k,2^{k}}\ne B^{(i)}_{k,0}=x_0$, the procedure described so far generates a multivalued process when consecutive reset intervals are concatenated. Therefore, once the multiresolution algorithm has been applied until the desired level of resolution, we obtain the proper discretization of the process setting  $B^{(i)}_{k,2^{k}}=x_0$, for $i=0,1,\ldots$ .

\subsection{Reflecting boundary}\label{subsubsec:4-1-1}

We now add the hard wall reflecting boundary~\cite{redner2001guide}
to the previous developments. The hard wall reflections that bound the process below level 1 annihilate the Gaussian nature of the process, which is essential to the multiresolution method. 
To solve this issue, we first generate an unbounded trajectory with the multiresolution method $\mathbf{B}_k^{(i)}$, at some resolution 
level $k \ge 1$. 
Then, by using the increments of $\mathbf{B}_k^{(i)}$, we define the reflected process $\mathbf{R}^{(i)}_{k}$ starting as follows:
\begin{gather}\label{eq:refcond}
\begin{split}
R^{(i)}_{k,j+1} & = \begin{cases}
R^{(i)}_{k,j}+\Delta B^{(i)}_{k,j},     & \text{ if } R^{(i)}_{k,j}+\Delta B^{(i)}_{k,j} \le 1 , \\
2 - (R^{(i)}_{k,j}+\Delta B^{(i)}_{k,j}), & \text{ if } R^{(i)}_{k,j}+\Delta B^{(i)}_{k,j} > 1 
\end{cases} \\
& = \mathrm{min}\left\{ R^{(i)}_{k,j} + \Delta B^{(i)}_{k,j}, 2- (R^{(i)}_{k,j} + \Delta B^{(i)}_{k,j}) \right\},
\end{split}
\end{gather}
where $\Delta B^{(i)}_{k,j} = B^{(i)}_{k,j+1}-B_{k,j}^{(i)}$,
for $j=0,\ldots,2^{k+1}$, and $R^{(i)}_{k,0}=x_0$.

\subsection{Stopping condition}

Every resolution $k$ will provide a value for the first time  at which the trajectory $\mathbf{R}_k$ is observed to go below the absorbing boundary, which is the null line,

\begin{equation} \label{eq:stopcond_V2}
    \tau_{r,k} = \inf \{t^{(i)}_{k,j}\in(t^\dag_{i-1},t^\dag_i) \, \vert \, R^{(i)}_{k,j} \leq 0 , \, i=1,2,\dots \} .
\end{equation}

This time provides an approximation to the target FPT at resolution $k$. Increasing the resolution yields finer sample paths and therefore finer estimations according to Eq.~\eqref{eq:stopcond_V2}. However, we note that estimations using Eq.~\eqref{eq:stopcond_V2} are always upper bounds to the real FPT. Furthermore, if $k>k'$, then 
\begin{equation*}
    \tau_{r}\le\tau_{r,k}\le\tau_{r,k'}.
\end{equation*}
This inequality has important consequences for our simulation schemes. Basically, every estimation induces an effective time horizon for our Monte Carlo method. Given the estimation at some resolution $k$, $\tau_{r,k}$, it makes no sense to simulate the process on times $t>\tau_{r,k}$ since $ \tau_{r}\le\tau_{r,k}$. 

An explicit condition at which we may stop the simulation can be defined by using an error threshold $\epsilon>0$. We define the maximum resolution level $k^{\dag}$ such that the time increment is below $\epsilon$,$$ k^{\dag} = - \left\lceil \frac{\log(\epsilon)}{\log(2)}\right\rceil, $$
where $\lceil x \rceil$ denotes the smallest integer
$\ge x$.

We finally defined all elements required to use the MRA to sample the FPT, we denote this version of the algorithm as the \textit{standard} MRA (SMRA). We start generating the unbounded motion from Eq.~\eqref{eq:unbounded_Br} in the first reset interval ($i=1$) for a target resolution $k^{\dag}$ using Eq.~\eqref{eq:brown_V2}.  Once the Brownian trajectory has been generated in the first reset interval, we compute the reflected process in this interval [through Eq.~\eqref{eq:refcond}]. If the stopping condition in the first interval is reached [Eq.~\eqref{eq:stopcond_V2}], then the simulation stops and we can sample an estimation for the FPT at resolution $k^{\dag}$. Otherwise, we will draw the next reset time, $t^\dag_2$, and proceed similarly in the second reset interval (i.e. generating the Brownian trajectory using the multiresolution scheme, then reflecting the trajectory and lastly checking if the reflecting trajectory reach the stopping criteria). This procedure is iterated until the stopping condition is reached.

Once the above process has been used to get an estimation of the FPT with resolution $k^{\dag}$, $\tau_{r,k^{\dag}}$, it can be further refined to reach a new resolution level $k > k^{\dag}$. In this second phase, there are no further sampling of reset times, and the multiresolution scheme is iterated on a fixed interval $[0,\tau_{r,k^{\dag}}]$. This is because $\tau_{r,k^{\dag}}$ is an upper bound for the true FPT.

\subsection{HMRA}
\label{subsec:hybrid}

Stochastic resetting can strongly increase the computational requirements of the SMRA. This is because it requires applying the multiresolution method in multiple reset intervals with a high resolution $k^\dag$, which is especially costly when $1/r \ll \tau_r$ because it results in a large number of resets.

\begin{figure}
    \centering
    \includegraphics[width=0.85\linewidth]{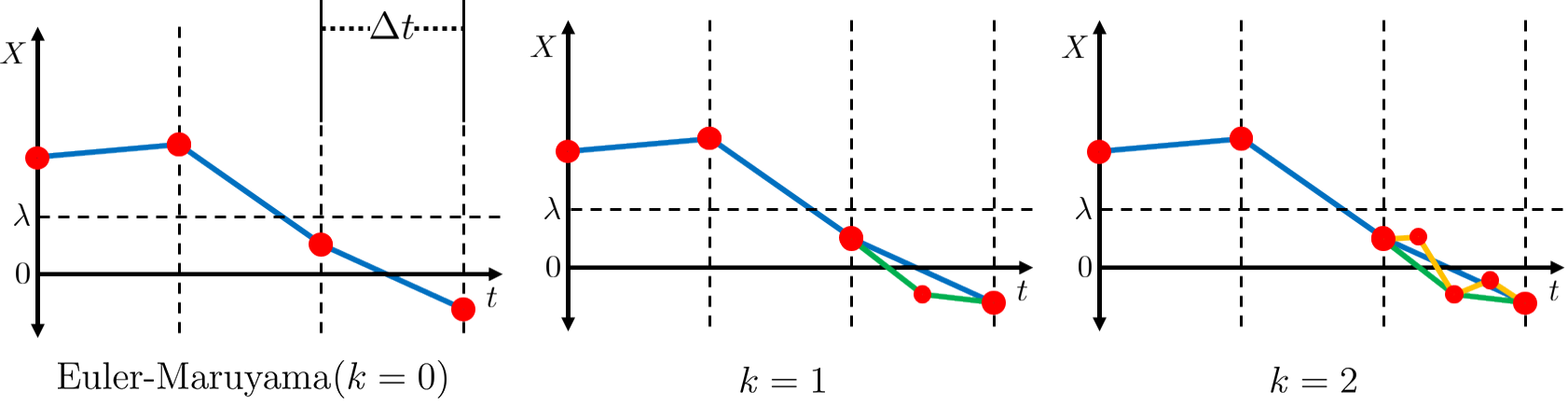}
    \caption{\it Illustrations of trajectories generated from HMRA. Time increments of the Euler-Maruyama algorithm ($\Delta t$) are sketched as the separation of vertical dashed lines.}
    \label{fig:muleul}
\end{figure}

The Euler-Maruyama algorithm partially overcomes these limitations, however it tends to overestimate the FPT and it does not have an arbitrary accuracy.
We propose the \textit{hybrid} MRA (HMRA) that refines the Euler-Maruyama trajectories with the MRA (Figure \ref{fig:muleul}). In doing so, first we produce a trajectory of the Euler-Maruyama (Algorithm \ref{alg:eulermaruyama}), with a time step $\Delta t$. Close to the absorbing boundary, i.e. $X<\lambda$, for some small $\lambda > 0$, we use the MRA to refine the approximation. The details are included in Appendix \ref{sec:appendix4}. Note finally that 
with our process
the higher order Milstein scheme
reduces to the Euler-Maruyama,
because the coefficient of $d W_t$ in Eq.
(\ref{eq:sde}), $\sqrt{2D}$, does not depend on $X_t$.

\section{Numerical results}                         \label{num} 

In order to illustrate the effectiveness of the PPF and the SMRA and HMRA, we analyse their performance in terms of accuracy, memory requirements, and speed. Source codes and Python packages of the PPF and MRA are available on the links provided in Appendices \ref{ap:proof} and \ref{ap:mra}.

\subsection{PPF}

The results of the approximation at order $m=2$ and $n=3$ are shown in Figures \hyperref[fig:distsim_compare]{3a} and \hyperref[fig:distsim_compare]{3b} which are compared with simulated results of the HMRA.
In Figure \ref{fig:distsim_compare} and Table \ref{table:distv0} we observe that the PPF method is a good approximation of the entire distribution obtained by Monte Carlo, for multiple values of $v$, except for the percentile $p=0.1$, for negative values of $v$.

\begin{figure}[!ht]
    \centering
    \includegraphics[width=0.9\linewidth]{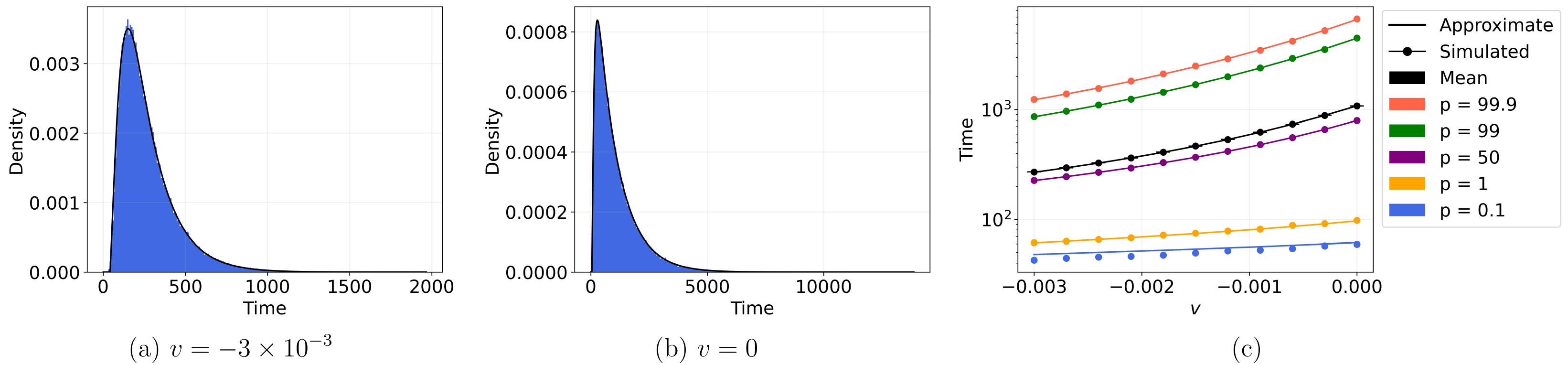}
    \caption{\it Results of PPF approximation, parameters are $D = 5 \times 10^{-4}$, $r = (3 \times 365)^{-1}$ and approximation orders $m = 2$, $n = 3$. (a)-(b) Histograms obtained from $10^6$ generated values and continuous curves obtained from PPF with varying drift $v = -3 \times 10^{-3}$ and $v=0$. (c) Comparing quantiles close to the tails, medians, and means generated from PPF and $10^6$ simulations for a varying drift.}
    \label{fig:distsim_compare}
\end{figure}

\begin{table}[!ht]
   \centering
    {\footnotesize
    \begin{tabular}{|l||l|l|}
    \hline
        $t$ & $\mathrm{P}_{HMRA}[\tau_r>t]$ & $\mathrm{P}_{PPF}[\tau_r>t]$ \\ \hline \hline
        342.140 & 0.25000 & 0.24918 \\ \hline
        379.322 & 0.19937 & 0.19800 \\ \hline
        416.503 & 0.15806 & 0.15718 \\ \hline
        453.685 & 0.12593 & 0.12470 \\ \hline
        490.866 & 0.10099 & 0.09891 \\ \hline
        528.047 & 0.07981 & 0.07843 \\ \hline
        565.229 & 0.06382 & 0.06219 \\ \hline
        602.410 & 0.05082 & 0.04931 \\ \hline
        639.592 & 0.04064 & 0.03910 \\ \hline
        676.773 & 0.03234 & 0.03100 \\ \hline
        713.955 & 0.02555 & 0.02458 \\ \hline
        751.136 & 0.02020 & 0.01949 \\ \hline
        788.317 & 0.01591 & 0.01545 \\ \hline
        825.499 & 0.01253 & 0.01225 \\ \hline
        862.680 & 0.01000 & 0.00971 \\ \hline
    \end{tabular}
    }
    {\footnotesize
    \begin{tabular}{|l||l|l|}
    \hline
        $t$ & $\mathrm{P}_{HMRA}[\tau_r>t]$ & $\mathrm{P}_{PPF}[\tau_r>t]$ \\ \hline \hline
        1445.966 & 0.25000 & 0.25089 \\ \hline
        1664.281 & 0.19842 & 0.19894 \\ \hline
        1882.596 & 0.15738 & 0.15775 \\ \hline
        2100.911 & 0.12428 & 0.12508 \\ \hline
        2319.225 & 0.09894 & 0.09919 \\ \hline
        2537.540 & 0.07823 & 0.07865 \\ \hline
        2755.855 & 0.06189 & 0.06236 \\ \hline
        2974.170 & 0.04967 & 0.04945 \\ \hline
        3192.485 & 0.03972 & 0.03921 \\ \hline
        3410.800 & 0.03117 & 0.03109 \\ \hline
        3629.115 & 0.02444 & 0.02465 \\ \hline
        3847.430 & 0.01964 & 0.01955 \\ \hline
        4065.745 & 0.01589 & 0.01550 \\ \hline
        4284.060 & 0.01255 & 0.01229 \\ \hline
        4502.375 & 0.01000 & 0.00975 \\ \hline
    \end{tabular}
    }
    \caption{\it Upper tail probabilities of HMRA and PPF in Figures \hyperref[fig:distsim_compare]{3a} (left table) and \hyperref[fig:distsim_compare]{3b} (right table).}
    \label{table:distv0}
\end{table}

 For a valid result from the PPF method, we need to identify all the roots of the denominator of the Pad{\'e} approximation, and ensure that they are real and negative, cf. Step 3 of Algorithm \ref{prop:Pade_approximant}. Having at least one positive root, implies that no Laplace inverse of $f_{m,n}$ exists. On the other hand, if there is at least one non-real root, the result yields a damped sinusoidal, i.e. a signed function. Varying the parameters for drift and diffusion reveals regions at which the PPF method will not work based on these criteria, as shown in Figure \ref{fig:rootvalid}.

\begin{figure}[!htb]
    \centering
    \includegraphics[width=0.85\linewidth]{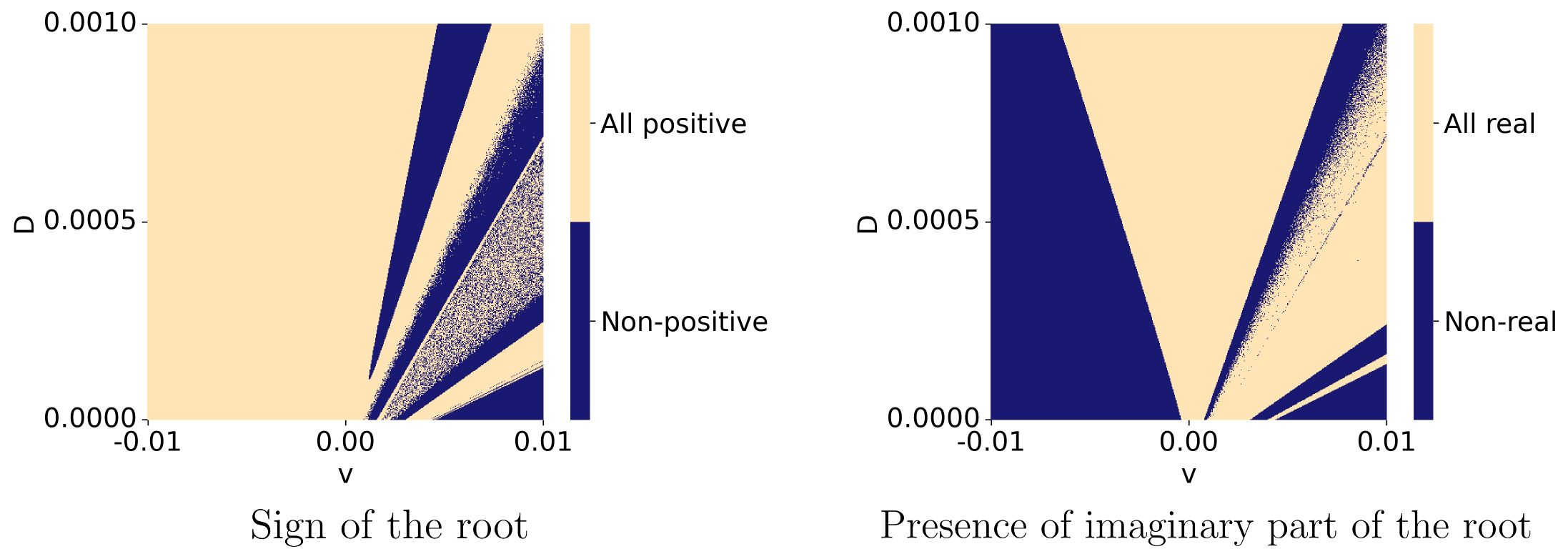}
    \caption{\it Regions of the parametric space $(v,D)$ of validity of the PPF with orders $m=2$, $n=3$, $r = (1/3) \times (1/365)$.}
    \label{fig:rootvalid}
\end{figure}

\subsection{SMRA and HMRA}

When comparing the accuracy of SMRA and HMRA by comparing their simulated mean FPT with the analytical mean FPT, we observe that as error threshold $\epsilon$ decreases, both algorithms converge to the analytical solution, whilst very small step sizes are necessary for Euler Maruyama to reach the same level of accurarcy (Figure \ref{fig:pvsh}).
In terms of computational requirements, when we decrease $\epsilon$, we observe a linear increase in $\langle k \rangle$, that would lead to an exponential increase in memory requirements as it is proportional to the number of points generated for each trajectory. Therefore, HMRA outperforms SMRA since it needs lesser points to reach the same accuracy (Figure \hyperref[fig:pareto]{6a}). Finally, in terms of speed, the HMRA outperforms the MRA by 2 to 4 orders of magnitude (Figure \hyperref[fig:pareto]{6b}).

\begin{figure}[!htb]
    \centering
    \includegraphics[scale = 0.4]{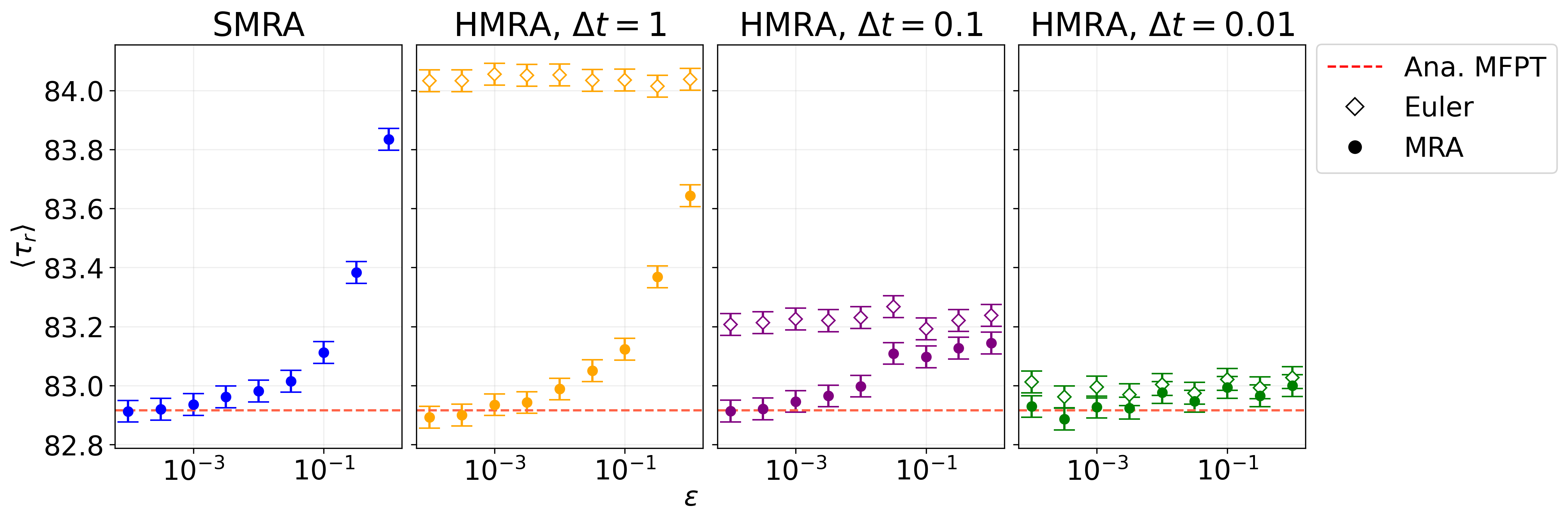}
    \caption{\it SMRA and HMRA results with different Euler-Maruyama time steps $\Delta t$, with changing threshold $\epsilon$. $v = -10^{-2}$, $D = 10^{-4}$, $r = (3\cdot 365)^{-1}$, $x_0 = 0.8$, $10^6$ simulations.}
    \label{fig:pvsh}
\end{figure}

\begin{figure}
    \centering
    \includegraphics[width=0.95\linewidth]{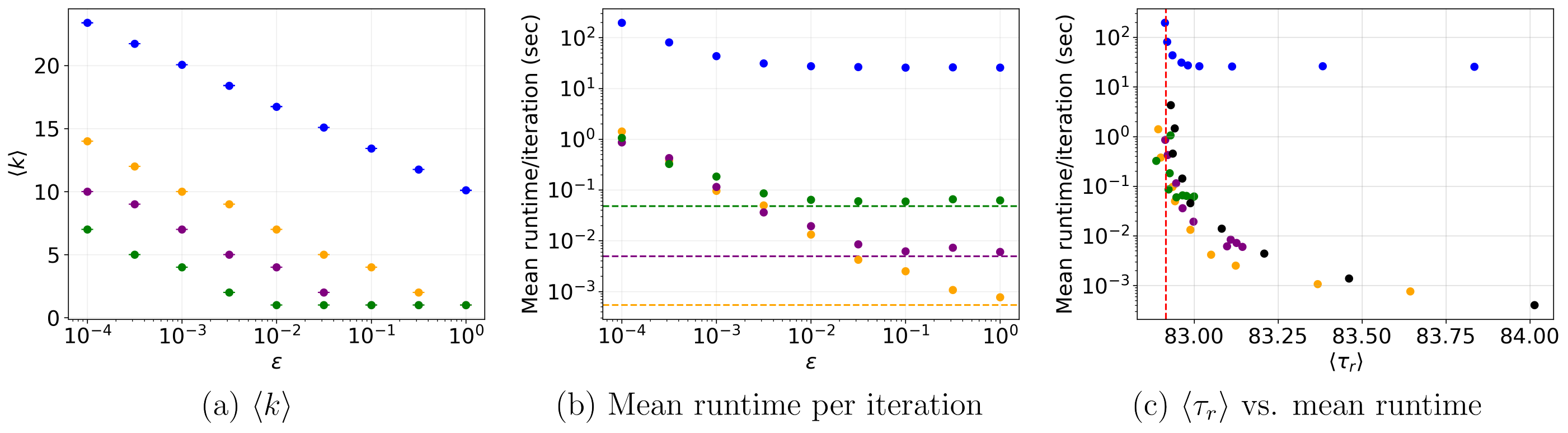}
    \caption{\it Speed and accuracy comparisons between the SMRA ({\color{blue} $\bullet$}), HMRA with $\Delta t = 1$ ({\color{orange} $\bullet$}), HMRA with $\Delta t = 0.1$ ({\color{Purple} $\bullet$}), and HMRA with $\Delta t = 0.01$ ({\color{Green} $\bullet$}). (a) Mean resolution per iteration. (b) Mean runtime per iteration. Dashed lines refer to the mean runtimes of the corresponding first Euler-Maruyama estimates for HMRA. (c) Comparisons between the mean FPT and the mean runtime per iteration of different simulation schemes. SMRA and HMRA simulations taken from Figure \ref{fig:pvsh}, Euler-Maruyama ($\bullet$) plotted simulations with varying $\Delta t$, vertical red dashed line is the analytical mean FPT.}
    \label{fig:pareto}
\end{figure}

Although HMRA outperforms the SMRA, HMRA requires two parameters to be set, $\epsilon$ and $\Delta t$. For this reason, we study the impact of either parameters in terms of computational time. For all the values of $\Delta t$, we identify where reducing $\epsilon$ does not correspond to a significant increase in computational time. However, as $\epsilon$ is further decreased, the computational time rapidly increases by orders of magnitude. This behavior describes a Pareto front, suggesting the existence of a critical value of $\epsilon$ which yields the best trade-off between accuracy and speed (Figure \hyperref[fig:pareto]{6c}).

\section{Discussion} \label{sec:discussion}

We have addressed the problem of computing the FPT distribution to the null level of the drifted Brownian motion with upper hard wall barrier and Poisson resetting times. In doing so, we have first introduced the PPF approximation, which is an analytical formula that can be immediately evaluated. We have then introduced the SMRA and HMRA,
that are purely numerical but give arbitrary accuracy, and we have shown how they overcome the limitations of some other available methods in terms of accuracy and speed. 
Moreover,
by using the survival function of the FPT, we have found a more compact expression of the mean FPT with resetting, which only uses the Laplace transform of the FPT without resetting \cite{evans_stochastic_2020, pal_first_2019}. We have provided an easier derivation than one in \cite{ramoso_stochastic_2020}.

Precisely, with the PPF we have proposed an approximation of the Laplace transform by taking the Pad{\'e} approximation and its partial fraction decomposition \cite{Davies1979}. These steps followed by exact inversion yield our PPF approximation.
It is accurate and, given that it does not rely on numerical integration, as do most methods, like the Talbot approximation \cite{talbot1979accurate, abateUnifiedFrameworkNumerically2006}, it yields an instantaneous numerical evaluation. 
The disadvantage of the PPF is that the region of the parametric space where the approximation is limited. But we have identified the region where it is valid.

To overcome this limitation, we have presented the MRA. Given two consecutive points on a trajectory, a property of Wiener processes is that the intermediate of the two points is Gaussian with mean given by the average of the first and last points \cite{Lawler_Gregory_F_1995-07-01}. The MRA exploits this bridge property by generating intermediate points between intervals of the trajectory. Hence, the MRA can generate finer values up to any threshold of error, compared to other methods for which the level of resolution is fixed in advance \cite{asmussen2007stochastic}, making the control of the error difficult.

We have introduced two versions of the MRA, namely the SMRA and HMRA. Either simulation algorithm has shown high accuracy and convergence,
this with respect to the exact analytical mean FPT and 
with respect 
to the simulation results of the standard Euler-Maruyama algorithm,
which is based on a single and very small time increment. 
It is known that compute the FPT
by Euler-Maruyama gives overestimation proportional to the time increment \cite{asmussen2007stochastic}. Thus, both the SMRA and HMRA correct this overestimation through the construction of Brownian bridges. 
Because
the SMRA computes the entire sample path over a large time horizon, 
(in order to obtain the FPT to an arbitrary accuracy), it requires substantial computational resources. To increase speed and memory efficiency, the HMRA proceeeds as follows: it first approximates the FPT with Euler-Maruyama and then it refines locally the resolution through MRA.
This yields a more efficient Monte Carlo technique than either the Euler-Maruyama or the SMRA individually.

Our Monte Carlo methods can readily be used with other Gaussian processes, such as the Ornstein-Uhlenbeck process
(see e.g. p. 229-230 of \cite{gatto2022stationary,dubey2023first}).
Beyond the FPT, our method could be used to estimate 
the first passage area, namely
the area enclosed between the null line and the path of the process up to the FPT, that has recently received substantial
attention, cf. e.g. \cite{abundo2023first}. 
We envisage that it would be possible to extend our results to non-Gaussian $\alpha$-stable L\'evy processes using generalizations of the Brownian bridge~\cite{Majumdar_2015_effective,Aguilar_PRE,Aguilar_Gatto_PRE}. Lastly, future work could focus on characterizing the optimal parameter choices that balance CPU time consumption and precision, as has been done for jumping processes~\cite{aguilar2024biased}.

\section{Code availability statement}

The source code (in Python) for the PPF, SMRA and HMRA is available on Github through the following links:
\begin{table}[!htb]
\centering
\begin{tabular}{rl}
PPF:           & \href{https://github.com/jarmsmagalang/ppf_approx}{\tt https://github.com/jarmsmagalang/ppf\textunderscore approx} \\
SMRA and HMRA: & \href{https://github.com/jarmsmagalang/multires}{\tt https://github.com/jarmsmagalang/multires} 
\end{tabular}
\end{table}

\noindent Furthermore, corresponding Python packages can be installed from PyPI at:
\begin{table}[!htb]
\centering
\begin{tabular}{rl}
PPF:           & \href{https://pypi.org/project/ppf-approx}{\tt https://pypi.org/project/ppf-approx} \\
SMRA and HMRA: & \href{https://pypi.org/project/multires}{\tt https://pypi.org/project/multires} 
\end{tabular}
\end{table}

\newpage

\centerline{\small \bf Acknowledgements}
\noindent 
{\small The authors are thankful to David Ginsbourger 
for his dedicated support and
insightful scientific discussions.

This work had been supported by the 
UniBE ID Grant 2021,
the Ruth \& Arthur Scherbarth Foundation,  
the Agencia Estatal de
Investigación (AEI, MCI, Spain) MCIN/AEI/10.13039/ 501100011033 and Fondo
Europeo de Desarrollo Regional (FEDER, UE) under Project APASOS
(PID2021$\textbf{-}$ 122256NB$\textbf{-}$C21$/$C22) and the María de Maeztu Program for units of
Excellence in $\text{R\&D}$, grant CEX2021$\text{-}$ 001164$\text{-}$M.}

\bibliographystyle{elsarticle-num} 
\bibliography{main}

\newpage

\begin{appendices}

\section{Proofs}\label{ap:proof}

\begin{proof}[Proof of Lemma~\ref{lemma:propagator}]
We begin by taking the Laplace transform of Eq. \eqref{eq:fke} with respect to $t$ and proceed to solve for the propagator using a procedure similar to the method of undetermined coefficients. 
Denote by $\mathcal{L}$ the operator of Laplace transform.
First, note that swapping integral and derivative operation over different variables yields
\begin{equation*}
	\mathcal{L}\left(vp'(x,\cdot)\right)(s)
	=v\tilde{p}'(x,s)\,
	\quad\text{and}\quad
	\mathcal{L}\left(Dp''(x,\cdot)\right)(s)
	=D\tilde{p}''(x,s)\,.
\end{equation*}

The first term is the only term in the equation that has a derivative dependent on $t$, which yields
\begin{equation*}
\partial_t p(x,t) = s\tilde{p}(x,s) - p(x,0) = s \tilde{p}(x,s) - \delta(x-x_0)
\end{equation*}
where $p(x,0) = \delta(x-x_0)$, taken from the initial condition in Eqs. \eqref{eq:system_of_p}. Therefore, the transform of Eq. \eqref{eq:fke}
for  $\tilde{p}(x,s)$ is given by a nonhomogeneous differential equation
\begin{equation}
\label{eq:fke_lt}
	s\tilde{p}(x,s) - \delta(x-x_0)
	+v\tilde{p}'(x,s)
	-D\tilde{p}''(x,s)
	=0\,,\quad s\geq 0\,,
\end{equation}
for $x\in[0,x_0)\cup(x_0,1]$. Rearranging the equation to have the nonhomogeneous term on the right-hand side, we obtain
\begin{equation}\label{eq:nonhom}
	D\tilde{p}''(x)-v\tilde{p}'(x)-s\tilde{p}(x)=-\delta(x-x_0).
\end{equation}
In order to find a solution to Eq.~\eqref{eq:fke_lt},
let us fix $s$, denote $\tilde{p}_<(x)=\tilde{p}(x,s)$ for $x \in [0,x_0)$,
$\tilde{p}_>(x)=\tilde{p}(x,s)$ for $x \in(x_0,1]$ and search an analytic solution to the associated homogeneous differential equation to Eq. \eqref{eq:nonhom}
\begin{equation*}
	D\tilde{p}''(x)-v\tilde{p}'(x)-s\tilde{p}(x)=0,
\end{equation*}
over these two intervals.
The characteristic equation is
$
	D\alpha^2-v\alpha-s=0$,
with roots, defined in Eq.~\eqref{eq:roots}, 
\begin{equation*}
	\alpha_\pm
	=\frac{v\pm\sqrt{v^2+4Ds}}{2D}
	=\frac{v\pm\omega(s)}{2D}
    =\rho\pm\theta(s)\,,
\end{equation*}
where the constants are defined by Eq.~\eqref{eq:constants}. Therefore,
\begin{equation*}
\tilde{p}_>(x)=Ae^{\alpha_+ x} + Be^{\alpha_- x}
\quad\text{and}\quad
\tilde{p}_<(x)=ae^{\alpha_+ x}+be^{\alpha_- x}\,,
\end{equation*}
for some constants $a$, $b$, $A$ and $B$ to be determined.
From the boundary conditions Eq. \eqref{eq:bc} we obtain
\begin{equation*}
	\tilde{p}_<(x)=a(e^{\alpha_+ x}-e^{\alpha_- x})
\end{equation*}
and
\begin{equation*}
	v(Ae^{\alpha_+}+Be^{\alpha_- }) = D(\alpha_+ Ae^{\alpha_+} +\alpha_- Be^{\alpha_-})\,.
\end{equation*}
The latter equation may be rewritten as
\begin{equation*}
	(D\alpha_+ - v)Ae^{\alpha_+} = (v - D\alpha_-)Be^{\alpha_-} = C(s)\,,
\end{equation*}
implying
\begin{equation*}
	\tilde{p}_>(x)=C\left(
	\frac{e^{\alpha_+ (x-1)}}{D\alpha_+ - v}+
	\frac{e^{\alpha_- (x-1)}}{v - D\alpha_-}
	\right)\,.
\end{equation*}
The continuity of $p(\cdot,t)$ (see \cite{karlin_taylor_1981}, Chapter 15, Section 5)
yields $\tilde{p}_<(x_0)=\tilde{p}_>(x_0)$, hence one can write
\begin{equation*}
	C\left(
	\frac{e^{\alpha_+ (x_0-1)}}{D\alpha_+ - v}+
	\frac{e^{\alpha_- (x_0-1)}}{v - D\alpha_-}
	\right)
	=
	a\left(
	e^{\alpha_+ x_0}-e^{\alpha_- x_0}
	\right)
	=c(s).
\end{equation*}
Consequently, we obtain
\begin{equation*}
	\tilde{p}_<(x)=c(s)
	\frac{e^{\alpha_+ x}-e^{\alpha_- x}}
	{e^{\alpha_+ x_0}-e^{\alpha_- x_0}}
	\; \text{ and } \;
	\tilde{p}_>(x)=c(s)
	\frac{
	\frac{e^{\alpha_+ (x-1)}}{D\alpha_+ - v}+
	\frac{e^{\alpha_- (x-1)}}{v - D\alpha_-}
	}
	{
	\frac{e^{\alpha_+ (x_0-1)}}{D\alpha_+ - v}+
	\frac{e^{\alpha_- (x_0-1)}}{v - D\alpha_-}
	}\,.
\end{equation*}
Noticing
\begin{equation*}
	v-D\alpha_-=D\alpha_+
	\quad\text{and}\quad
	v-D\alpha_+=D\alpha_-,
\end{equation*}
we have
\begin{equation*}
	\tilde{p}_>(x)=c(s)
	\frac{\alpha_+ e^{\alpha_+ (x-1)} - \alpha_- e^{\alpha_- (x-1)}}
	{\alpha_+ e^{\alpha_+ (x_0-1)} - \alpha_- e^{\alpha_- (x_0-1)}}\,.
\end{equation*}
Our goal is now to determine the function $c(s)=\tilde{p}(x_0,s)$.

We proceed to examine conditions involving the derivative of $p$ following standard procedures (see e.g. p. 16
of \cite{redner2001guide} or p. 449 of \cite{Arfken7th}). We now return to the nonhomogeneous differential equation in Eq. \eqref{eq:nonhom}. Integrating the above equation on a segment around $x_0$,
\begin{align*}
	\int_{x_0-\epsilon}^{x_0+\epsilon}\left\{ D\tilde{p}''(x,s) - v\tilde{p}'(x,s)
	-  s\tilde{p}(x,s)\right\} \, dx
	&=-\int_{x_0-\epsilon}^{x_0+\epsilon}\delta(x-x_0) \, dx \,,\quad s\geq 0\,, \epsilon>0 
 \end{align*}
 then,
 \begin{align*}
 D\tilde{p}'(x,s) \Big|_{x_0-\epsilon}^{x_0+\epsilon} - v\tilde{p}(x,s) \Big|_{x_0-\epsilon}^{x_0+\epsilon} - \int_{x_0-\epsilon}^{x_0+\epsilon}\left(s\tilde{p}(x,s)\right) \, dx \
	&=-1.
\end{align*}
In the limit $\epsilon \rightarrow 0$, the last two terms at the left-hand side tend to zero due to continuity of $\tilde{p}(x,s)$ at $x = x_0$. While the first summand can be rearranged and evaluated in terms of $\tilde{p}'_<$ and $\tilde{p}'_>$,

\begin{equation}
\label{eq:mistery}
	D\left( \tilde{p}'_<(x_0)-\tilde{p}'_>(x_0)	\right)=1\,.
\end{equation}

Let us compute the first derivative
\begin{equation*}
	\tilde{p}'_<(x)=c \frac{\alpha_+ e^{\alpha_+ x} -\alpha_- e^{\alpha_- x}}{e^{\alpha_+ x_0}-e^{\alpha_- x_0}}\,,
\end{equation*}
and evaluate it at $x_0$
\begin{align*}
	\tilde{p}'_ < (x_0)&=c\frac{\alpha_+ e^{\alpha_+ x_0}	-\alpha_- e^{\alpha_- x_0}}{e^{\alpha_+ x_0}-e^{\alpha_- x_0}}
	\\
	&=c \frac{\alpha_+ e^{\theta x_0} -\alpha_- e^{-\theta x_0}}{e^{\theta x_0}-e^{-\theta x_0}}	\,.
\end{align*}
Analogously,
\begin{equation*}
\tilde{p}'_>(x)=c\frac{\alpha_+^2 e^{\alpha_+ (x-1)} - \alpha_-^2 e^{\alpha_- (x-1)}}{\alpha_+ e^{\alpha_+ (x_0-1)} - \alpha_- e^{\alpha_- (x_0-1)}}\,,
\end{equation*}
leads to
\begin{align*}
	\tilde{p}'_ >(x_0)&=c\frac{\alpha_+^2 e^{\alpha_+ (x_0-1)} - \alpha_-^2 e^{\alpha_- (x_0-1)}}{\alpha_+ e^{\alpha_+ (x_0-1)} - \alpha_- e^{\alpha_- (x_0-1)}}
	\\
	&=c\frac{\alpha_+^2e^{\theta (x_0-1)}-\alpha_-^2e^{-\theta (x_0-1)}}{\alpha_+ e^{\theta (x_0-1)}-\alpha_- e^{-\theta (x_0-1)}}\,.
\end{align*}
Therefore, by inserting Eq.~\eqref{eq:mistery}
\begin{align*}
	1&=Dc\left(
    \frac{\alpha_+ e^{\theta x_0}-\alpha_- e^{-\theta x_0}}{e^{\theta x_0}-e^{-\theta x_0}}
	-
	\frac{\alpha_+^2e^{\theta (x_0-1)}-\alpha_-^2e^{-\theta (x_0-1)}}{\alpha_+ e^{\theta (x_0-1)}-\alpha_- e^{-\theta (x_0-1)}}
    \right)
	\\
	&=Dc \, \frac{e^{-\theta}(\alpha_+^2-\alpha_+\alpha_-)+e^\theta(\alpha_-^2-\alpha_+\alpha_-)}{\alpha_+ e^{\theta(2x_0-1)}+ \alpha_- e^{-\theta(2x_0-1)}-\alpha_+ e^{-\theta} -\alpha_- e^\theta}
	\\
	&=2Dc \, \frac{e^{-\theta}\theta(\theta+\rho)+e^\theta\theta(\theta-\rho)}{\alpha_+ e^{\theta(2x_0-1)}+ \alpha_- e^{-\theta(2x_0-1)}-\alpha_+ e^{-\theta} -\alpha_- e^\theta}
	\\
	&=2Dc \, \frac{\theta (\theta \, \mathrm{cosh}(\theta)-\rho \, \mathrm{sinh}(\theta))}{\rho \, \mathrm{cosh}[\theta (2x_0-1)]+\theta \, \mathrm{sinh}[\theta (2x_0-1)] -\rho \, \mathrm{cosh}(\theta)+\theta \, \mathrm{sinh}(\theta)}
	\\
	&=c\, \frac{\omega (\omega \, \mathrm{cosh}(\theta)-v \, \mathrm{sinh}(\theta))}{v \, \mathrm{cosh}[\theta (2x_0-1)]+\omega \, \mathrm{sinh}[\theta (2x_0-1)] -v\, \mathrm{cosh}(\theta)+\omega \, \mathrm{sinh}(\theta)}\,,
\end{align*}
which gives
\begin{align*}
	c(s)=\tilde{p}(x_0, s)
	&=\frac{v \, \mathrm{cosh}[\theta(2x_0-1)]+\omega \, \mathrm{sinh}[\theta(2x_0-1)]-v \, \mathrm{cosh}(\theta)+\omega \, \mathrm{sinh}(\theta)}{\omega^2 \, \mathrm{cosh}(\theta)-v\omega \, \mathrm{sinh}(\theta)}
	\\
	&=\frac{2 \, \mathrm{sinh}(\theta x_0)\left\{\omega \, \mathrm{cosh}[\theta(x_0-1)]+v \, \mathrm{sinh}[\theta(x_0-1)]\right\}}{\omega^2 \, \mathrm{cosh}(\theta)-v \, \omega\mathrm{sinh}(\theta)}\,,
\end{align*}
where $\omega$ and $\theta$ are defined in Eq.~\eqref{eq:constants}.
All in all, we obtain the expression Eqs. \eqref{eq:p_below}
and \eqref{eq:px_0}
for the Laplace transform of the propagator. 
\end{proof}

\begin{proof}[Proof of Proposition~\ref{p22}]
\label{AP:proof_22}
Let us define the distribution function of $\tau_r$ by
$F_r(t)=\int_0^t f_r(s)ds$ and its
\emph{survival function} by 
$S_r(t)=1-F_r(t)=\PP(\tau_r>t)$, $t\geq 0$.
In view of the known relation
\begin{equation}
\label{eq:kr}
	f_r(t)=-\partial_t{S}_r(t)\,,
\end{equation} 
we may re-express the expectation of interest in the following way:
\begin{equation}
\label{eq:mat0}
	 \E [\tau_r]
    =\int_0^\infty t f_r (t)\,dt
	=-\int_0^\infty t \partial_t{S}_r (t)\,dt
	=\int_0^\infty S_r(t)\,dt
	= \tilde{S}_r(0)\,,
\end{equation}
where the tilde always denotes the Laplace transform of a function.
In order to find an expression for the latter term, we use a renewal equation that connects the survival function in the resetting case to
the survival function for the evolution without resetting.
The following equation is well-known in the literature relative to
stochastic resetting and can be found e.g. in \cite{evans_stochastic_2020,pal_first_2019, roldan_path-integral_2017}
\begin{equation*}
	S_r(t)=e^{-rt}S_0(t) + r\int_0^t e^{-ru}S_0(u)S_r(t-u)\,du\,.
\end{equation*}
By taking Laplace transform on both sides we obtain
\begin{equation*}
	\tilde{S}_r(s)=\tilde{S}_0(s+r)+r \tilde{S}_0(s+r) \tilde{S}_r(s)\,,
\end{equation*}
hence, for all $s$ such that $\tilde{S}_0(s+r) \neq r^{-1}$, 
\begin{equation}
\label{eq:re_lt}
	\tilde{S}_r(s)=\frac{\tilde{S}_0(s+r)}{1-r\tilde{S}_0(s+r)}\,.
\end{equation}
The link between $\tilde{S}_0(s)$ and $\tilde{f}_0(s)$ is provided by
the Laplace transform of Eq.~\eqref{eq:kr}, which for $r=0$ yields, for all $s \neq 0$,
\begin{equation*}
	\tilde{S}_0(s)=\frac{1-\tilde{f}_0(s)}{s}\,.
\end{equation*}
Plugging the latter into Eq.~\eqref{eq:re_lt} gives for all $s \neq - r$ such that $\tilde{f}_0(s+r) \neq -s/r$,
\begin{equation} \label{eq:tildeSr}
	\tilde{S}_r(s)
	=\frac{1-\tilde{f}_0(s+r)}{s+r\tilde{f}_0(s+r)}\,,
\end{equation}
and by Eq. \eqref{eq:mat0} we finally obtain
\begin{equation}
\label{eq:mat1}
	\E [\tau_r] =\frac{1}{r}\left\{\frac{1}{\tilde{f}_0(r)}-1\right\}\,.
\end{equation}
All that remains to do is expressing $\tilde{f}$ in terms of
what we explicitly know, that is $\tilde{p}$. This is done by evaluating the Laplace transform of the probability current (Eq. \eqref{eq:probcurrent}) at the absorbing boundary $x=0$
\begin{equation}
	\tilde{f}_0(s) = \tilde{J}(0,s) = 
     D \tilde{p}'(0,s) - v \tilde{p}(0,s)
    = D \tilde{p}'(0,s) \, .
    \label{eq:249}
\end{equation}
The probability current $J(x,t)$ is the rate of influx of Brownian particles at the position $x$, and since $x=0$ is an absorbing boundary, $J(0,t)$ will give the FPT distribution \cite{redner2001guide}. This is related to the derivation of the inverse Gaussian distribution using the method of images \cite{Whitmore1987, Chhikara1988-nr}.
Therefore, from Eq. \eqref{eq:249} we obtain Eq. \eqref{eq:tildef0} for \\ $s>- v^2/(4D)$
\begin{equation}
    \tilde{f}_0(s)=
	e^{-\rho x_0}\frac{\omega(s) \, \mathrm{cosh}\{\theta(s)(x_0-1)\}+v \, \mathrm{sinh}\{\theta(s)(x_0-1)\}}
	{\omega(s) \, \mathrm{cosh}\{\theta(s)\}-v \, \mathrm{sinh}\{\theta(s)\}}\,,
\end{equation}
as it has been previously shown in \cite{Godec2016} and plugging the latter into Eq.~\eqref{eq:mat1} yields an explicit analytic expression for the mean absorption time
\begin{equation*}
	\mathbb{E}[\tau_r]
	=\frac{1}{r}\left\{\frac{e^{\rho x_0} \, (\omega(r) \, \mathrm{cosh}\{\theta(r)\}-v \,\mathrm{sinh}\{\theta(r)\})}{\omega(r) \, \mathrm{cosh}\{\theta(r)(x_0-1)\}+v \, \mathrm{sinh}\{\theta(r)(x_0-1)\}}-1\right\}\,,
\end{equation*} 
where $\rho$, $\omega$ and $\theta$ have been defined in Eq.~\eqref{eq:constants}.
\end{proof}

\newpage

\section{Implementation of SMRA and HMRA}\label{ap:mra}

\subsection{MRA}\label{AP:multiresolution_algorithm}
We first consider a simple case of the MRA in the time interval $[0,1]$. In the algorithm, the discretized Wiener process, $\{ W_{k,j} \}$, with $k,j$ nonnegative integers,
is obtained through Eqs. (\ref{eq:w00}) and
(\ref{eq:wiener}). 

\begin{algorithm}
\caption{Basic MRA}
\label{alg:basicmultires}
\KwIn{initial position $W_{0}$}
\KwOut{$\mathbf{W} = \{W_{0}, \ldots, W_{2^{k-1}}  \}, \mathbf{t} = \{0, \ldots, t^\dag  \}$}
$k \coloneqq 0$ \;
$W_1 \coloneqq W \sim \mathcal{N}(0,1)$ \;
$\mathbf{W} \coloneqq \{W_{0}, W_{1}\}$ \;
$\mathbf{t} \coloneqq \{0,1 \}$ \;
\While{$\mathbf{W}$ has not reached a stopping condition}{
$k \coloneqq k+1$ \;
$\mathbf{W}^* \coloneqq \emptyset; \quad \mathbf{t}^* \coloneqq \emptyset$ \;
\For{$j \in [0,2^{k}]$}
{
$j' \coloneqq \left \lfloor \dfrac{j}{2} \right \rfloor$ \;
\uIf{$j$ is even}{
$W_j^* \coloneqq W_{j'}$ \Comment*[r]{note: $W_{j'}$ is $j'$-th element of $\mathbf{W}$}
$t_j^* \coloneqq t_{j'}$ \Comment*[r]{note: $t_{j'}$ is $j'$-th element of $\mathbf{t}$}
}
\uElse{
$\mu_{int} \coloneqq \dfrac{W_{j'}+W_{j'+1}}{2}$ \;
$\sigma^2_{int} \coloneqq \dfrac{1}{2^{k}}$ \;
$W_{j}^* \coloneqq W \sim \mathcal{N}(\mu_{int}, \sigma^2_{int})$\;
$t_j^* \coloneqq \dfrac{t_{j'}+t_{j'+1}}{2}$ \;
}
}
$\mathbf{W} \coloneqq \mathbf{W}^*$\;
$\mathbf{t} \coloneqq \mathbf{t}^*$ \;
}
\end{algorithm}

This pseudocode will generate the two arrays $\mathbf{W}$ and $\mathbf{t}$, 
the positions and times respectively of the trajectories. 
The final resolution level $k$ can be arbitrarily obtained 
with a given stopping condition.

\newpage
\subsection{Euler-Maruyama trajectory for HMRA}
\label{sec:appendix2}

This method is similar to the Euler-Maruyama algorithm described in Algorithm \ref{alg:eulermaruyama}, but modified for use in the hybrid algorithm. This code will output two arrays $\mathbf{X}^E$ and $\mathbf{t}^E$ which are Euler-Maruyama trajectories but below a position threshold $\lambda>0$. This is because above this certain threshold in the position, an absorption is unlikely to happen. Note that reflection and resetting occur in this algorithm.

\vspace{5mm}

\begin{algorithm}[H]
\caption{\texttt{euler}: Generate Euler trajectory below threshold $\lambda$}
\LinesNumbered
\KwIn{$X_0,t_0,v,D,r,\Delta t, \lambda$}
\KwOut{$\mathbf{X}^E,\mathbf{t}^E$}
$X \coloneqq X_0$ \;
$t \coloneqq t_0$ \;
$\mathbf{X}^E \coloneqq \emptyset$ \;
$\mathbf{t}^E \coloneqq \emptyset$ \;
$t^\dag \coloneqq t_r \sim \mathrm{Exp}(1/r)$ \;
\While{$X>0$}{
\uIf(\Comment*[f]{resetting condition}){$t \geq t^\dag$}{
$X \coloneqq X_0$ \;
$ t \coloneqq t^\dag + t_r, \quad \text{where } t_r \sim \mathrm{Exp}(1/r)$\;
}
\uElse{
$X \coloneqq X + \Delta X, \quad \text{where }  \Delta X \sim \mathcal{N}(v \, \Delta t, 2D \, \Delta t)$ \;
$t \coloneqq \min \{t+ \Delta t, t^\dag \}$ \;
\uIf(\Comment*[f]{reflection condition}){$X \geq 1$}{
$X \coloneqq 2- X$ \;
}
}
\uIf{$X < \lambda$}{
$\mathbf{append}$ $X$ to array $\mathbf{X}^E$ \;
$\mathbf{append}$ $t$ to array $\mathbf{t}^E$ \;
}
}
\end{algorithm}

\newpage
\subsection{Euler-Maruyama algorithm}
\begin{algorithm}[!htb]
\caption{Euler-Maruyama algorithm for FPT} \label{alg:eulermaruyama}
\KwIn{$X_0,v,D,r,\Delta t$}
\KwOut{$\tau_r$}
$X \coloneqq X_0$\;
$t \coloneqq 0$ \;
$t^\dag \coloneqq t_r \sim \mathrm{Exp}(1/r)$ \;
\While{$X>0$}{
\uIf(\Comment*[f]{resetting condition}){$t \geq t^\dag$}{
$X \coloneqq X_0$ \;
$t \coloneqq t^\dag + t_r, \quad \text{where } t_r \sim \mathrm{Exp}(1/r)$ \;
}
\uElse{
$X \coloneqq X + \Delta X, \quad \text{where } \Delta X \sim \mathcal{N}(v \, \Delta t, 2D \, \Delta t)$ \;
$t \coloneqq \min \{t+ \Delta t, t^\dag \}$ \;
\uIf(\Comment*[f]{reflection condition}){$X \geq 1$}{
$X \coloneqq 2- X$ \;
}
\uElseIf(\Comment*[f]{first passage/stopping condition}){$X \leq 0$}{$\tau_r \coloneqq t$ \;}
}
}
\end{algorithm}

\newpage
\subsection{SMRA}
\label{sec:appendix3}

This code uses the $\texttt{multires}$ function found in Appendix \ref{AP:multiresolution_algorithm}. This algorithm generates a Brownian trajectory $\mathbf{X}_k$, $\mathbf{t}_k$ and outputs its corresponding FPT with the absorbing boundary up to an error threshold $\epsilon$.

Recall the simulation parameters $k^\dag$ discussed in Section \ref{subsubsec:4-1-2}. Parameter $k^\dag$ is the minimum resolution that the trajectory must have before resetting is allowed, while $k^*$ is the maximum resolution before the FPT is recorded and the algorithm is stopped. Note that $k^*>k^\dag$.

\vspace{5mm}

\begin{algorithm}[H]
\caption{SMRA}
\LinesNumbered
\KwIn{$x_0,v,D,r,\epsilon, k^*$}
\KwOut{$\tau_r$}
$t_0 \coloneqq 0$ \;
$t^\dag \coloneqq t' \sim \mathrm{Exp}(1/r)$ \;
$k^{\dag} \coloneqq - \left\lceil \dfrac{\log(\epsilon)}{\log(2)}\right\rceil$ \;
\While{$\delta_k>\epsilon \, \mathbf{or} \, k<k^*$}
{
$k \coloneqq 0$ \;
$B_{f} \coloneqq B' \sim \mathcal{N}(x_0+v(t^\dag-t_0),2D(t^\dag-t_0))$ \;
$\mathbf{B} \coloneqq \{B_{0}, B_{f}\}$ \;
$\mathbf{t} \coloneqq \{t_0,t^\dag \}$ \;
$\mathbf{X} \coloneqq \{x_0\}$ \;
\While{all $X \in \mathbf{X}> 0 \, \mathbf{or} \, k<k^\dag$}
{
$k \coloneqq k+1$ \;
$\delta_k \coloneqq \dfrac{t^\dag}{2^k}$ \;
$\mathbf{B}, \mathbf{t} \coloneqq \texttt{multires}(\mathbf{B}, \mathbf{t}, D, k, t^\dag)$ \;
\For(\Comment*[f]{reflection condition}){$j = 0 \; \mathbf{to} \; 2^k$}{
$\Delta B_j \coloneqq B_{j+1}-B_j$ \;
$X_{j+1} \coloneqq \mathrm{min}\left\{ X_{j} + \Delta B_{j}, \, 2- (X_{j} + \Delta B_{j}) \right\}$  \;
}
\uIf(\Comment*[f]{stopping condition}){any $X \in \mathbf{X} < 0 \, \mathbf{and} \, (\delta_k<\epsilon \, \mathbf{or} \, k>k^*)$}{
$\tau_r = \inf \{t \in \mathbf{t} \, \vert \, X_{t}<0 \}$ \;
$\mathbf{break}$ loops in lines $\mathbf{3}$ and $\mathbf{9}$
}
\uElseIf(\Comment*[f]{resetting condition}){$k>k^\dag$}
{
$t_0 \coloneqq t^\dag$ \;
$t^\dag \coloneqq t^\dag + t', t' \sim \mathrm{Exp}(1/r) $ \;
$\mathbf{break}$ loop in line $\mathbf{9}$ and $\mathbf{return}$ to line $\mathbf{3}$
}
\uElse(\Comment*[f]{increase resolution}){$\mathbf{return}$ to line $\mathbf{9}$}
}
}
\end{algorithm}

\newpage

\subsection{HMRA}
\label{sec:appendix4}

This code uses both the $\texttt{multires}$ function from Appendix \ref{AP:multiresolution_algorithm} and the $\texttt{euler}$ function from Appendix \ref{sec:appendix2}. The algorithm begins by generating an Euler-Maruyama trajectory $\mathbf{X}^E$ and $\mathbf{t}^E$ below the position threshold $\lambda$. Note that the reflections and resets have occurred in the initial Euler-Maruyama trajectory already.
\vspace{2mm}

\begin{algorithm}[H]
\caption{HMRA}
\LinesNumbered
\KwIn{$X_0,v,D,r,\epsilon, \Delta t, k^*, \lambda$}
\KwOut{$\tau_r$}
$t_0 \coloneqq 0$ \;
$k \coloneqq 0$ \;
$\mathbf{X}^E, \mathbf{t}^E \coloneqq \texttt{euler}(X_0,t_0,v,D,r,\Delta t, \lambda)$ \;
$\mathbf{X}^L \coloneqq \emptyset$ \;
$\mathbf{t}^L \coloneqq \emptyset$ \;
\For{$i = 0 \; \mathbf{to} \; \text{length of array } \mathbf{X}^E$}{
\uIf{$(t^E_{i+1}-t^E_i) \leq \Delta t$}{
$X_i^L \coloneqq \{X_i^E, X_{i+1}^E\}$ \;
$t_i^L \coloneqq \{t_i^E, t_{i+1}^E\}$ \;
}
}
\While{$\delta_k < \epsilon$}{
$k \coloneqq k +1$ \;
$\delta_k \coloneqq \dfrac{t^\dag}{2^k}$ \;
\For{$\ell = 0 \; \mathbf{to} \; \text{length of array } \mathbf{X}^L$}{
$X' \coloneqq X^L_\ell$ \Comment*[r]{note: $X^L_\ell = \{X^E_i, \ldots,  X^E_{i+1}\}_\ell$}
$t' \coloneqq t^L_\ell$ \Comment*[r]{note: $t^L_\ell = \{t^E_i, \ldots,  t^E_{i+1}\}_\ell$}
${X}^L_\ell, t^L_\ell \coloneqq \texttt{multires}(X', t', D, k, t^\dag)$ \;
}
$\mathbf{X}^M \coloneqq \text{flattened array of } \mathbf{X}^L$ \;
$\mathbf{t}^M \coloneqq \text{flattened array of } \mathbf{t}^L$ \;
\uIf(\Comment*[f]{stopping condition}){any $X \in \mathbf{X}^M < 0 \; \mathbf{and} \; (\delta_k<\epsilon \; \mathbf{or} \; k>k^*)$}
{
$\tau_r = \inf \{t \in \mathbf{t}^M \, \vert \, X^M_t<0 \}$ \;
$\mathbf{break}$ loop $\mathbf{10}$
}
\uElse(\Comment*[f]{increase resolution}){$\mathbf{return}$ to line $\mathbf{10}$}

}
\end{algorithm}
\vspace{2mm}
The loop at line 6 splits both $\mathbf{X}^E$ and $\mathbf{t}^E$ into an array of arrays $\mathbf{X}^L$ consisting consecutive elements of the original array, e.g. for $\mathbf{X}^E$: $\mathbf{X}^L = \{\{X^E_0, X^E_1\}, \{X^E_1, X^E_2\}, \ldots \}$. Each array element of both $\mathbf{X}^L$ and $\mathbf{t}^L$ is passed through $\texttt{multires}$ and afterwards, the array of arrays is flattened back to a 1D array which is called $\mathbf{X}^M$ and $\mathbf{t}^M$, to check for the first passage and the stopping condition.

\newpage

\end{appendices}

\end{document}